\documentclass[twocolumn]{autart}
\usepackage{graphicx}
\usepackage{amsmath,amssymb}
\usepackage{listings}
\usepackage{color}
\usepackage{url}
\usepackage{booktabs}

\usepackage{algorithm,algorithmic}
\usepackage{graphicx,subcaption,calc}

\newtheorem{theorem}{Theorem}
\newtheorem{lemma}{Lemma}
\newtheorem{proof}{Proof}

\usepackage{esvect}

\usepackage{algorithm,algorithmic}
\usepackage{savesym}
\savesymbol{AND}%

\newcommand{\Exp}[1]{\operatorname{E}\left[#1\right]}
\newcommand{\ExpO}{\operatorname{E}}

\newcommand{\Cov}[1]{\operatorname{Cov}\left[#1\right]}

\newcommand{\vecV}[1]{\operatorname{vec}\left(#1\right)}
\newcommand{\vecVh}[1]{\operatorname{vech}\left(#1\right)}

\renewcommand\mid{\,\vert\,}

\newcommand{\Tr}{\operatorname{Tr}}
\newcommand{\R}{\mathbb{R}}
\newcommand{\N}{\mathcal{N}}

\newcommand{\GP}{\mathcal{GP}}
\newcommand{\Transp}{\mathsf{T}}

\newcommand{\bmat}[1]{\begin{bmatrix}#1\end{bmatrix}}

\newcommand{\argmax}[2]{\underset{#1}{\mathrm{arg}\,\mathrm{max}}\,\,#2}

\newcommand{\myd}{\textrm{d}}      

\begin{document}

\begin{frontmatter}
\title{Stochastic quasi-Newton with line-search regularization}
\thanks[footnoteinfo]{This paper was not presented at any IFAC meeting. Corresponding author: A. Wills.}

\author[Newcastle]{Adrian Wills}\ead{Adrian.Wills@newcastle.edu.au}, 
\author[Uppsala]{Thomas B. Sch\"on}\ead{thomas.schon@it.uu.se}
\address[Newcastle]{University of Newcastle, School of Engineering, Callaghan, NSW 2308, Australia.}
\address[Uppsala]{Uppsala University, Department of Information Technology, 751 05 Uppsala, Sweden.}


\begin{abstract}                          
  In this paper we present a novel quasi-Newton algorithm for use in
  stochastic optimisation. Quasi-Newton methods have had an enormous
  impact on deterministic optimisation problems because they afford rapid
  convergence and computationally attractive algorithms. In essence,
  this is achieved by learning the second-order (Hessian) information
  based on observing first-order gradients. We extend these ideas to
  the stochastic setting by employing a highly flexible model for the
  Hessian and infer its value based on observing noisy gradients. In
  addition, we propose a stochastic counterpart to standard
  line-search procedures and demonstrate the utility of this
  combination on maximum likelihood identification for general
  nonlinear state space models.
\end{abstract}

 \end{frontmatter}


\section{Introduction}
%
We are interested in the non-convex stochastic optimisation problem
\begin{align}
	\label{eq:5}
	\min_{x} f(x),
\end{align}
where we only have access to noisy evaluations of the cost function $f(x)$ and its first-order derivatives. The problem has a long history and an important landmark development was the so-called \textit{stochastic approximation} idea derived by Robbins and Monro almost 70 years ago \cite{RobbinsM:1951}. The central idea of stochastic approximation is to form a Markov-Chain for $x$ via
\begin{align}
  \label{eq:1}
  x_{k+1} &= x_k + \alpha_k p_k
\end{align}
that converges (under fairly mild conditions) to a local minimum of \eqref{eq:5} for careful choices of the \emph{search-direction} $p_k$ and \emph{step-length} $\alpha_k > 0$ (see e.g. \cite{BottouCN:2018}).

%
In recent years the relevance of this problem has massively increased mainly due to the fact that it arises in at least the following two important situations. First, when the cost function and its gradients are intractable, but where we can still make use of numerical methods to compute noisy estimates (preferably unbiased) of these objects. Second, when the cost function inherently depends on a very large amount of data and it becomes impractical to evaluate it or the associated gradient using the entire dataset. It is then standard to use only a smaller fraction of the data, which is commonly referred to as minibatching. This situation arises in large-scale application of supervised machine learning and in particular in deep learning.

%
The application that motivated us to develop the solution presented in this work is that of maximum likelihood identification of nonlinear state space models. This is a specific instance of the first situation mentioned above. The cost function---the likelihood---is intractable, but we have numerical methods---sequential Monte Carlo (SMC) \cite{Gordon:1993,Kitagawa:1993,StewartM:1992}---that provide unbiased estimates of the likelihood~\cite{Delmoral:2004,PittSGK:2012}. For two relatively recent overviews of the use of SMC---a.k.a. particle filters---in nonlinear system identification, see \cite{SchonLDWNSD:2015,Kantas:2015}.

%
Our \textit{main contribution} is a new stochastic optimisation algorithm which features mechanisms facilitating the use of second-order information (Hessian) in calculating the search-direction $p_k$, and a stochastic line search to  compute the step-length $\alpha_k$. The representation used for the Hessian is provided by the Gaussian process \cite{RasmussenW:2006} and we develop a method for its online updating as the optimisation algorithm progresses. We also derive a stochastic line-search procedure that employs a version of the Armijo condition \cite{Armijo:1966,Wolfe:1969,Wolfe:1971}. It is perhaps surprising that little work has been done when it comes to developing stochastic line-search methods \cite{BollapragadaMNST:2018}. We stress that while the developments mentioned above are generally applicable, the application that motivated us to undertake this work is the nonlinear system identification problem, which provides an important spin-off contribution.



\section{Related work}
\label{sec:rw}


Due to its importance, the stochastic optimisation problem is rather well studied by now. The first stochastic optimisation algorithm was introduced in \cite{RobbinsM:1951}. It makes use of first-order information only, motivating the name stochastic gradient (SG), which is the contemporary 
term \cite{BottouCN:2018} for these algorithms, originally referred to as stochastic approximation. Interestingly most SG algorithms are not descent methods since the stochastic nature of the update can easily produce a new iterate corresponding to an increase in the cost function. Instead, they are Markov chain methods in that their update rule defines a Markov chain. 


Since the landmark paper \cite{RobbinsM:1951}, many extensions and modifications have been developed within the statistics and automatic control communities. While the contributions are too many to enumerate here, some notable works include convergence results~\cite{kiefer1952stochastic,ljung1977analysis,ljung1978strong}, online parameter estimation and system identification in~\cite{Ljung:1979,Ljung:1983}, adaptive control strategies~\cite{goodwin1981discrete}, and general books in the area~\cite{bertsekas1996neuro,spall2005introduction,ljung2012stochastic}. It is important to note that a primary motivation for the current paper is the closely related area of system identification for nonlinear dynamic systems. 


This existing research is having an enormous impact in the related area of machine learning at present and we believe the reason is simple: it is used to solve almost all supervised machine learning problems, including all deep learning problems~\cite{BottouCN:2018}. This is evidenced by all available toolboxes in the area offering SG algorithms and variants of them. Due to this, SG methods are still receiving enormous research attention.


The primary focus of current research activity is directed towards producing algorithms with improved convergence rates. Two important aspects that impact convergence rate are:
\begin{itemize}
\item Poor problem scaling can lead to slow convergence~\cite{BottouCN:2018};
\item Classical step-length formulas are conservative~\cite{AsiD20129siam,MoulinesB:2011}.
\end{itemize}


Regarding the first point, a well-known drawback of first-order methods is that the choice of coordinate system can greatly impact the rate of convergence, which is highlighted in the concluding comments of a recent review paper on the topic~\cite{BottouCN:2018}. Incorportating second-order information (Hessian) can alleviate the sensitivity to coordinate choice, which is one of the motivations for the Newton's method and it's locally quadratic convergence rate. At the same time, in many pratical situations, computing the Hessian is impractical, which is certainly the case in many system identification problems and almost all deep learning problems. Addressing this problem are the infamous suite of quasi-Newton methods such as the BFGS method \cite{Broyden:1970,Fletcher:1970,Goldfarb:1970,Shanno:1970}, Broyden's method~\cite{Broyden:1965} and the DFP formula \cite{FletcherP:1963,Broyden:1967}. In essence, these algorithms learn the Hessian (or its inverse) matrix based on first-order information, resulting in fast convergence and computationally attractive algorithms. 


For the stochastic setting of the current paper, these classical quasi-Newton methods are not applicable~\cite{BottouCN:2018}. Towards addressing this, over the past decade we have witnessed increasing capabilities of so-called \emph{stochastic quasi-Newton methods}, the category to which our current developments belong.  The work by \cite{SchraudolphYG:2007} developed modifications of BFGS and its limited memory version. There has also been a series of papers approximating the inverse Hessian with a diagonal matrix, see e.g. \cite{BordesBG:2009} and \cite{DuchiEY:2011}. The idea of exploiting regularisation together with BFGS was successfully introduced by \cite{MokhtariR:2014}. Some of these approaches rely on the assumption that maintaining common random numbers between two gradient evaluations will result in locally deterministic behaviour. While this assumption may be satisfied for certain classes of functions, it is not valid of nonlinear system identification situation when SMC based calculation of the cost and gradient are employed~\cite{Kantas:2015}. 


In the current paper, we take a different approach and develop a new quasi-Newton algorithm that deals with the stochastic problem directly. The development stems from stating the integral form of the so-called quasi-Newton equation, where an exact relation between gradient differences and the Hessian are provided. This provides a natural point for treating the stochastic gradients and we then propose the use of a flexible model structure for the unknown Hessian matrix based on formulating the problem as a Gaussian Process regression problem. This results in a fully probabilistic model for the Hessian, where the mean function is used as a surrogate Hessian matrix in a new quasi-Newton algorithm.


Regarding the second point, the step-length schedule is also addressed in the current paper by suggesting a stochastic line-search procedure modelled after the backtracking line-search with Armijo conditions~\cite{Armijo:1966,Wolfe:1969,Wolfe:1971}. It is interesting---and perhaps somewhat surprising---to note that it is only very recently that stochastic line search algorithms have started to become available. One nice example is the approach proposed by \cite{MahsereciH:2017} which uses the framework of Gaussian processes and Bayesian optimisation. The step length is chosen that best satisfies a probabilistic measure combining reduction in the cost function with satisfaction of the Armijo condition. Conceptually more similar to our procedure is the line search proposed by \cite{BollapragadaMNST:2018}, which is tailored for problems that are using sampled mini-batches, as is common practice within deep learning. The final line-search algorithm proposed in the current paper begins with a more classical backtracking style procedure and converges towards a deterministic schedule that satisfies the typical convergence requirements for SG methods.

While we started this development in our earlier conference paper \cite{WillsS:2017}, that paper missed several key ingredients that we provide in this paper. In particular we have now developed a correct and properly working mechanism for representing and updating the local Hessian using a GP. Furthermore, we have introduced a completely new and improved line search algorithm.


\section{Stochastic quasi-Newton method}
Quasi-Newton methods have found enormous success within the field of optimisation~\cite{NocedalW:2006}. The primary reasons are that they capture the cost function curvature information and that they are computationally inexpensive, requiring only gradient calculations. The curvature information leads to better scaling of the negative gradient direction and consequently to faster convergence~\cite{NocedalW:2006}. 

In pursuit of similar advantages for the stochastic optimisation
problem, many authors have considered how to develop quasi-Newton
algorithms when the gradient vector is stochastic (see e.g. the
concluding remarks in \cite{BottouCN:2018}). While the potential
benefits of incorporating curvature information are widely recognised,
it remains a rapidly evolving area.

In order to describe our quasi-Newton approach in
Section~\ref{sec:LocalHessian}, here we provide some background
material with a slightly non-standard introduction to the quasi-Newton
method (Section~\ref{sec:SQN_intro}). This leads to the so-called
\textit{quasi-Newton integral} from which the more classic
\textit{secant equation} can be obtained. The secant equation is
essential to all standard quasi-Newton methods such as BFGS, DFP and
SR1 methods, and is obtained by introducing a rather strong
approximation. We refrain from making this approximation and in
Section~\ref{sec:SQN_stochastic} we show that we can formulate a
stochastic version of the quasi-Newton integral.
\subsection{A non-standard quasi-Newton introduction}
\label{sec:SQN_intro}
%

%
The idea underlying the Newton and quasi-Newton methods is to
\emph{learn a local quadratic surrogate model} $q(x_k, \delta)$ of the
cost function $f(x)$ around the current iterate $x_k$
\begin{align}
  \label{eq:22}
  q(x_k,\delta) \triangleq f(x_k) + \delta^{\Transp}\nabla f(x_k) +
  \frac{1}{2} \delta^{\Transp} \nabla^2 f(x_k) \delta,
\end{align}
where $\delta = x - x_k$. Note that~\eqref{eq:22} is a second-order
Taylor expansion of $f(x)$, i.e.  $f(x) \approx q(x_k,\delta)$ in a
close vicinity around the current iterate~$x_k$.

Quasi-Newton methods are used in situations where the curvature
information inherent in the Hessian is important at the same time as
it is too expensive or impossible to compute the Hessian. The idea is
to introduce a representation for the Hessian and compute an estimate
of the Hessian using zero- and first-order information (function
values and their gradients).  More specifically the existing
quasi-Newton methods are designed to represent the cost function
according to the following model
\begin{align}
\label{eq:16a}
f_{\text{q}}(x_k+\delta) &= f(x_k) + \delta^{\Transp} \nabla f(x_k) + \frac{1}{2} \delta^{\Transp} H_k \delta,
\end{align}
for some matrix $H_k$ representing the Hessian. 

In order to see how to update the Hessian approximation as the
algorithm proceeds let us assume that $x_k$ and $x_{k+1}$ are known. 
The line segment $r_k(\tau)$ connecting these two iterates can be expressed as
\begin{align}
\label{eq:qN:LineSegment}
r_k(\tau) = x_k + \tau(x_{k+1} - x_k), \qquad \tau\in [0, 1].
\end{align}
The fundamental theorem of calculus states that 
\begin{align}
\int_0^1 \frac{\partial }{\partial \tau} \nabla f(r_k(\tau)) \myd\tau
\notag
&= \nabla f(r_k(1)) - \nabla f(r_k(0)) \\
&= \nabla f(x_{k+1}) - \nabla f(x_{k}) \label{eq:qN:FTC}
\end{align}
and the chain rule allows us to express the integrand as
\begin{align}
\notag
\frac{\partial }{\partial \tau} \nabla f(r_k(\tau)) &= \nabla^2 f(r_k(\tau)) \frac{\partial r_k(\tau)}{\partial \tau} \\
\label{eq:qN:chainRule}
&= \nabla^2 f(r_k(\tau)) (x_{k+1} - x_{k}).
\end{align}
By inserting~\eqref{eq:qN:chainRule} into~\eqref{eq:qN:FTC} we conclude that
\begin{align}
\label{eq:qN:int1}
\int_0^1 \nabla^2 f(r_k(\tau)) (x_{k+1} - x_{k}) \myd\tau = \nabla f(x_{k+1}) - \nabla f(x_{k}).
\end{align}
We introduce the following notation 
\begin{align}
	\label{eq:ysDef}
  s_k \triangleq x_{k+1} - x_{k},
\end{align}
for the differences between two adjacent iterates. The integral~\eqref{eq:qN:int1} can be written in a more convenient form according to
\begin{align}	
\label{eq:qN:int}
\nabla f(x_{k+1}) - \nabla f(x_{k}) = \left [ \int_0^1 \nabla^2 f(r_k(\tau)) \myd\tau \right ] \  s_k, 
\end{align}
which we will refer to as the \emph{quasi-Newton integral}. The
interpretation of the above development is that the difference between
two consecutive gradients constitute a \emph{line
  integral observation of the Hessian}. The challenge is that the
Hessian is unknown and there is no functional form available for it.

%
%
%
The approach taken by existing quasi-Newton algorithms is to assume
that the Hessian can be described by a zero-order Taylor expansion,
\begin{align}
\nabla^2 f(r_k(\tau)) \approx H_{k+1}, \qquad \tau\in[0, 1], 
\end{align}
implying that the integral equation in~\eqref{eq:qN:int} is approximated according to
\begin{align}
\label{eq:qN:secantCond}
\nabla f(x_{k+1}) - \nabla f(x_{k}) = H_{k+1} s_k.
\end{align}
This equation is known as the \emph{secant condition} or the
\emph{quasi-Newton equation}~\cite{Fletcher:1987,NocedalW:2006}. It
constrains the possible choices for the Hessian approximation
$H_{k+1}$ since in addition to being symmetric, it should at least
satisfy the secant condition~\eqref{eq:qN:secantCond}. This is still
not enough to uniquely determine the Hessian and we are forced to
include additional assumptions to find the Hessian
approximation. Depending on which assumptions are made we recover the
standard quasi-Newton algorithms. See~\cite{Hennig:2015} for
additional details.

Our approach is fundamentally different in that we will allow the
Hessian to evolve according to a general nonlinear function which is
represented using a Gaussian process.

\subsection{Stochastic quasi-Newton formulation}
\label{sec:SQN_stochastic}
%
%
\label{sec:SqN}
The situation we are interested in corresponds to the case where we
have access to noisy evaluations of the cost function and its
gradients. In particular, assume that the computed gradient $g_k$ can be
modelled as
\begin{align}
  \label{eq:2}
  {g}_k &= \nabla f(x_k) + v_k
\end{align}
and where we make the simplifying assumption that $v_k$ is independent
and identically distributed according to
\begin{align}
  \label{eq:3}
  v_k \sim \mathcal{N}(0,R), \qquad R \succ 0.
\end{align}
Therefore, via~\eqref{eq:ysDef}--\eqref{eq:qN:int}
\begin{align}
	\notag
  \nabla f(x_{k+1}) - \nabla f(x_{k}) &= \left [ \int_0^1 \nabla^2 f(r_k(\tau)) \myd\tau \right ] \  s_k \\
  \label{eq:StochQNint}  
      &= {g}_{k+1} - {g}_k + v_{k+1} - v_k,
\end{align}
so that if we define 
\begin{align}
  \label{eq:6}
  {y}_k &= {g}_{k+1} - {g}_k
\end{align}
then we obtain a \emph{stochastic quasi-Newton integral} according to 
\begin{align}
  \label{eq:StochQNint}
  {y}_k &= \left [ \int_0^1 \nabla^2 f(r_k(\tau))\myd\tau  \right ]s_k + w_k,
\end{align}
where $w_k = v_{k} - v_{k+1}$. Based on~\eqref{eq:StochQNint} the
Hessian can now be estimated via the gradients that we have
available. To enable this we first need a suitable representation for
the Hessian. However, before we make that choice in the subsequent
section, let us rewrite the integral slightly to clearly exploit the
fact that the Hessian is by construction a symmetric matrix.

%
%

Note that since $\nabla^2 f(r_k(\tau)) s_k$ is a column vector we can
straightforwardly apply the vectorisation operator inside the integral
in~\eqref{eq:StochQNint} without changing the result,
\begin{align}
	\notag
	\nabla^2 f(r_k(\tau)) s_k &= \vecV{ \nabla^2 f(r_k(\tau)) s_k } \\
	\notag
	&= (s^{\Transp}_k \otimes I) \vecV{\nabla^2 f(r_k(\tau))} \\
	\label{eq:GP:SqNIntVec}
	&= (s^{\Transp}_k \otimes I) \vecV{\nabla^2 f(r_k(\tau))},
\end{align}
where $\otimes$ denotes the Kronecker product and $\vecV{\cdot}$ is the
vectorisation operator that when applied to an $n \times m$ matrix $A$
produce a column vector by stacking each column $A_i \in \R^n$ for
$i=1,\ldots,m$ via
\begin{align}
  \label{eq:11}
  \vecV{A} &\triangleq \bmat{A_1 \\ \vdots\\A_m} \in \R^{nm \times 1}
\end{align}
The whole point of this exercise is that we have now isolated the
Hessian in a vectorised form~$\vecV{\nabla^2 f(r_k(\tau))}$ inside the
resulting integral
\begin{align}
  \label{eq:StochQNint2}
  {y}_k = (s^{\Transp}_k \otimes I) \int_0^1 \vecV{\nabla^2 f(r_k(\tau))} \myd\tau + w_k.
\end{align}
The so-called half-vectorisation operator\footnote{For a symmetric
  $n\times n$ matrix $A$, the vector $\vecV{A}$ contains redundant
  information. More specifically, we do not need to keep the
  $n(n-1)/2$ entries above the main diagonal. The half-vectorisation
  $\vecVh{A}$ of a symmetric matrix $A$ is obtained by vectorising
  only the lower triangular part of $A$.}
$\vecVh{\cdot}$~\cite{MagnusN:1980} is now a useful bookkeeping tool
to encode the fact that the Hessian is symmetric. It will effectively
extract the unique elements of the full Hessian and conveniently store
them in a vector for us according to
\begin{align}
	h(r_k(\tau)) = \vecVh{\nabla^2 f(r_k(\tau))}.     
\end{align}
We can then retrieve the full Hessian using the so-called \emph{duplication matrix} $D$, which is a matrix such that
\begin{align}
	\label{eq:halfVecOperator}
	\vecVh{\nabla^2 f(r_k(\tau))} = D h(r_k(\tau)).
\end{align}
More details and some useful results on the duplication matrix, the associated \emph{elimination matrix} $L$ ($\vecVh{A} = L \vecV{A}$) and their use are provided by~\cite{MagnusN:1980}. Finally, inserting~\eqref{eq:halfVecOperator}  into~\eqref{eq:StochQNint2} results in 
\begin{align}
  \notag
  {y}_k &= \underbrace{(s^{\Transp}_k \otimes I)D}_{= \bar{D}_k}
        \int_0^1 h(r_k(\tau)) \myd\tau + w_k\\
  \label{eq:SQNfinal}
      &= \bar{D}_k \int_0^1 h(r_k(\tau)) \myd\tau + w_k.
\end{align}
We have now arrived at an integral providing us with information about
the unique elements of the Hessian $\nabla^2 f(x)$ via the
noisy gradient observations in ${y}_k$.


\section{Local Hessian representation and learning}
\label{sec:LocalHessian}
%
%
If the cost function~$f(x)$ is truly quadratic, then the Hessian
$\nabla^2 f(x)$ is constant. On the other hand, in the more interesting
case when~$f(x)$ is not quadratic, the Hessian $\nabla^2 f(x)$ can
change rapidly and in a nonlinear way as a function of~$x$. Together
with the fact that our measurements of the cost function and it
gradients are noisy this motivates the need for a Hessian
representation that can accommodate nonlinear functions in a
stochastic setting.

There are of course many candidates for this, but the one we settle
for in this work is the \emph{Gaussian process}. It is a natural
candidate since it provides a non-parametric and probabilistic model
of nonlinear functions~\cite{RasmussenW:2006}. Besides these two
reasons we would also like to mention the fact that all existing
quasi-Newton algorithms can in fact be interpreted as maximum a
posteriori estimates where a very specific Gaussian prior was used for
the unknown variables. This was relatively recently discovered
by~\cite{Hennig:2015}. 

As a fourth and final motivation, the GP is also very simple to work
with since it only involves manipulations of multivariate Gaussian
distributions. An intriguing consequence of representing the Hessian
using a GP is that it opens up for new algorithms compared to the
previously available quasi-Newton algorithms. This was indeed also
mentioned as an avenue for future work in~\cite{Hennig:2015}, see
also~\cite{HennigK:2013}. However, apart from our very preliminary
work in~\cite{WillsS:2017} this is---to the best of our
knowledge---the first assembly of a working algorithm of this kind.

Motivated by the above we will assume a GP prior for the unique
elements of the Hessian according to
\begin{align}
	\label{eq:GPprior4Hessian}
	h \sim \GP(\mu, \kappa),
\end{align}
where $\mu$ denotes the mean value function and $\kappa$ denotes the
covariance function. With this in place we now need to solve a
slightly non-standard GP regression problem in order to incorporate
the gradient information into a posterior distribution of the Hessian
$p(h(x)\mid {y}_{k-p},\ldots, {y}_k)$, where
$\{{y}_i\}_{i=k-p}^k$ denotes the noisy gradient-difference
measurements that are available. The non-standard nature of this
regression problem is due to the integral that is present
in~\eqref{eq:SQNfinal}. However, the GP is closed under linear
operators~\cite{RasmussenW:2006}, and it is possible to make use of
the line integral formulation in~\eqref{eq:SQNfinal}. This will be
exploited to derive closed-form expression for the Hessian posterior
making use of the most recent $p+1$ observations available in the
gradients and iterates.

Since we will make frequent use of the most recent $p+1$ indecies, it is
conevient to define an index set $\ell_k$ as
\begin{align}
  \ell_k \triangleq \{ k-p,k-p+1, \ldots, k\}.
\end{align}
Using this notation, we introduce bold face notation for the stacked
gradient differences and iterate differences according to
\begin{subequations}
	\begin{align}
	\mathbf{{y}}_{\ell_k} &\triangleq \begin{pmatrix}
	{y}_{k-p}^{\Transp} & {y}_{k-p+1}^{\Transp} & \cdots & {y}_k^{\Transp}
	\end{pmatrix}^{\Transp},\\
	\mathbf{s}_{\ell_k} &\triangleq \begin{pmatrix}
	s_{k-p}^{\Transp} & s_{k-p+1}^{\Transp} & \cdots & s_k^{\Transp}
	\end{pmatrix}^{\Transp}.
	\end{align}
\end{subequations}
The joint distribution of the unknown~$h(x)$ and the
known~$\mathbf{{y}}_{\ell_k}$ is given by
\begin{align}
  \label{eq:GP:jointDist}
  \begin{pmatrix}
    h(x)\\
    \mathbf{{y}}_{\ell_k}
  \end{pmatrix} \sim \N\left(\begin{pmatrix}
      \mu(x)\\
      \mathbf{m}_{\ell_k}
    \end{pmatrix}, 
  \begin{pmatrix}
    \kappa(x,x) & K_{x, {\ell_k}}\\
    K_{{\ell_k} , x} & K_{{\ell_k}, {\ell_k}}
  \end{pmatrix}\right).
\end{align}
where the mean vector $\mathbf{m}_{\ell_k}$ is given by
\begin{align}
  \mathbf{m}_{\ell_k} \triangleq \begin{pmatrix} m_{k-p}^\Transp, \ldots m_k ^\Transp\end{pmatrix}^\Transp
\end{align}
and
\begin{align}
  \notag
  m_i &= \Exp{{y}_i} = \bar{D}_i \int_0^1
        \Exp{h(r_i(\tau))} \myd\tau \\
  \label{eq:GP:yMean}
      &= \bar{D}_i \int_0^1 \mu(r_i(\tau))\myd \tau.
\end{align}
The covariance matrix~$K_{{\ell_k}, {\ell_k}}$ is given by 
\begin{align}
  K_{{\ell_k}, {\ell_k}} \triangleq 
\begin{pmatrix} 
  K_{k-p, k-p} & K_{k-p, k-p+1}, & \cdots & K_{k-p, k} \\ 
  K_{k-p+1, k-p} & K_{k-p+1, k-p+1}, & \cdots & K_{k-p+1, k} \\ 
  \vdots &  & \ddots & \vdots\\ 
  K_{k, k-p} & K_{k, k-p+1}, & \cdots & K_{k, k} \\ 
\end{pmatrix}
\end{align}
with each matrix element $K_{i,j}$ given by
\begin{align}
K_{i,j} &= \Exp{\left({y}_i - \Exp{{y}_i}\right)\left({y}_j -
          \Exp{{y}_j}\right)^{\Transp}} \notag \\
	\notag 
	&= \ExpO\Bigg[\left( \bar{D}_i\int_{0}^{1}\left(h(r_i(\tau)) - \mu(r_i(\tau))\right)\myd\tau + w_i \right)\\
	\notag
	&\quad \times \left( \bar{D}_j\int_{0}^{1}\left(h(r_j(t)) - \mu(r_j(t))\right)\myd t + w_j \right)^{\Transp} \Bigg]\\
&= \bar{D}_i \int_{0}^{1}\int_{0}^{1} \kappa(r_i(\tau), r_j(t)) \myd \tau
  \myd t\bar{D}_j^{\Transp} + R\delta_{i,j}, \label{eq:GP:Cov1}
\end{align}
where $\delta_{i,j}$ reflects the covariance structure of (see~\eqref{eq:StochQNint})
\begin{align}
 \Exp{w_iw_j^\Transp} = \Exp{(v_{i} - v_{i+1}) (v_j - v_{j+1})^\Transp}, 
\end{align}
so that 
\begin{align}
\label{eq:14delta}
  \delta_{i,j} = 
\begin{cases} 
  2  & \text{if} \quad i=j,\\
  -1 & \text{if} \quad |i-j|=1,\\ 
  0 & \text{otherwise}.
\end{cases}
\end{align}
The terms in the cross-covariance $K_{x, {\ell_k}}$ are given by (and
analogously for~$K_{{\ell_k}, x}$)
\begin{align}
  K_{x, {\ell_k}} \triangleq 
\begin{pmatrix} 
  K_{x, k-p} & K_{x, k-p+1}, & \cdots & K_{x, k}
\end{pmatrix},
\end{align}
with each matrix element $K_{x,j}$ given by
\begin{align}
  K_{x, j} &= \Exp{ \left ( h(x) 
  - \mu(x)\right)\left({y}_j -
  \Exp{{y}_j}\right)^{\Transp}} \notag \\
&= \ExpO\Bigg[\left( h(x) 
  - \mu(x) \right) \notag\\ 
  &\quad \times \left( \bar{D}_j\int_{0}^{1}\left(h(r_j(t)) - \mu(r_j(t))\right)\myd t + e_j \right)^{\Transp} \Bigg] \notag \\
 &= \int_{0}^{1} \kappa(x, r_j(t)) \myd
  t\, \bar{D}_j^{\Transp}.
	\label{eq:GP:Cov2}
\end{align}
The posterior distribution can now be computed by applying the
standard result of conditioned Gaussian distributions to the joint
Gaussian distribution~\eqref{eq:GP:jointDist} resulting in
\begin{subequations}
  \label{eq:16}
  \begin{align}
    h(x) \mid \mathbf{{y}}_{\ell_k} &\ \sim \ \N\left(
                                    \phi_{\ell_k}(x),
                                    \Sigma_{\ell_k}(x)
                                    \right),
\end{align}
where
\begin{align}
	\phi_{\ell_k}(x) &= \mu(x) - K_{x,\ell_k}K_{\ell_k,
                           \ell_k}^{-1} (\mathbf{{y}}_{\ell_k} -
                           \mathbf{m}_{\ell_k}), \label{eq:19}\\ 
          \Sigma_{\ell_k}(x) &= \kappa(x,x) - K_{x,\ell_k}K_{\ell_k,
                           \ell_k}^{-1} K_{\ell_k, x}. \label{eq:20}
	\end{align}
\end{subequations}
To actually be able to work with this, the integrals
in~\eqref{eq:GP:Cov1} and~\eqref{eq:GP:Cov2} have to be
computed. These integrals do of course depend on our particular choice
of kernel. We make use of a multivariate version of the so-called
squared exponential kernel
\begin{align}
\label{eq:GP:SEkernel}
\kappa(x, x^{\prime}) = M\exp\left(-\frac{1}{2}(x - x^\prime)^\Transp V (x - x^\prime)\right),
\end{align}
where the positive definite symmetric matrix $M \succ 0$ describes the covariance effect on each element
of $h(\cdot)$ and the positive definite symmetric matrix~$V\succ 0$ acts as an inverse length
scale. The resulting integrals in~\eqref{eq:GP:Cov1}
and~\eqref{eq:GP:Cov2} are non-trivial, but they can be efficiently
computed using standard numerical tools, see~\cite{HendriksJWS:2018} for all the details.
When a different kernel is used, the alternatives are to either solve
the corresponding integrals or to employ the simplifying assumption we
develop in the subsequent section. 

We conclude this section by illustrating the above GP approach on a
simple one dimensional example. Consider the function
\begin{align}
  \label{eq:23}
  f(x) = 4(x-6)^2 + e^{1.2x-5} + 10 - 10\sin(1.2x).
\end{align}
We obtain $N=12$ noise corrupted observations of the gradient where $g_k
= \nabla f(x_k) + v_k$ and $v_k \sim \mathcal{N}(0,100)$. The iterates
$x_k \in [-5,7]$. We employ the squared-exponential kernel
\eqref{eq:GP:SEkernel} with $M=1000$ and $V=0.2$ for the covariance
$\kappa(x,x^\prime)$ and select the mean function as $\mu(x) =
100$. The resulting Hessian approximation in terms of the GP conditioned on $y_{1:N}$ and $s_{1:N}$ is shown in Figure~\ref{fig:hess_gp_fig} for a range of $x$ values. Note that in the region $x \in [-15,-5]$, we have not learned very much so the 
GP converges to the prior. Whereas, for the region where we observe 
the gradient, the resulting GP approximation supports the exact Hessian.

\begin{figure}[bth]
\begin{center}
  \includegraphics[width=0.95\columnwidth]{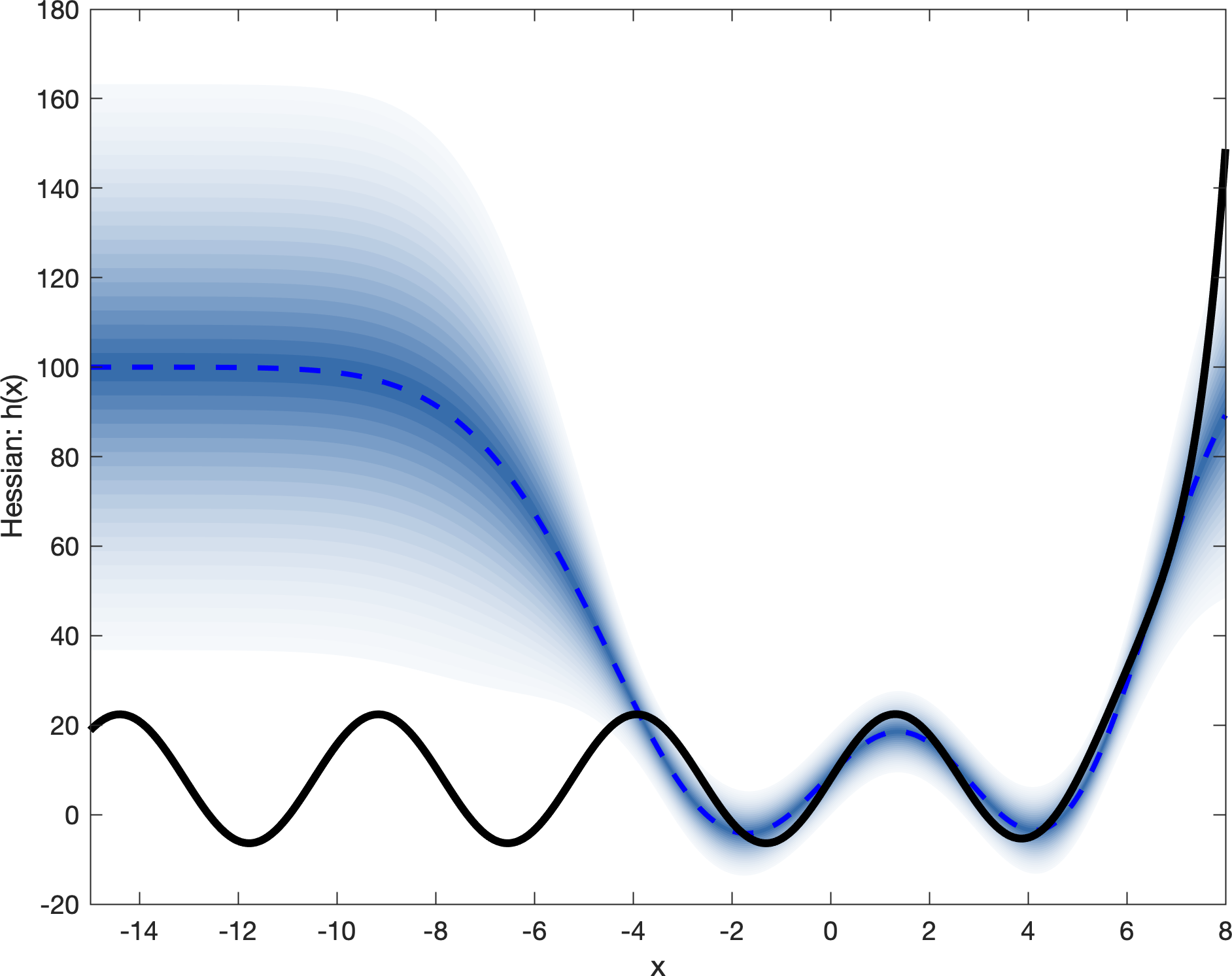}%
  \caption{GP Hessian approximation for the function~\eqref{eq:23}. Exact
    Hessian (solid black line), GP based Hessian mean (dashed blue)
    and covariance (shaded blue).}
  \label{fig:hess_gp_fig}
\end{center}
\end{figure}

\subsection{A simplifying approximation}
The above development can be simplified by invoking the standard
quasi-Newton approximation~\eqref{eq:qN:secantCond}, effectively
removing the integrals. Indeed, if the Hessian can be modelled as a
constant matrix between $x_k$ and $x_{k+1}$, then the
measurements can be simplified to
\begin{align}
  \label{eq:StochQNint2}
  {y}_k &= \left [ \int_0^1 \nabla^2 f(r_k(\tau))\myd\tau  \right
                  ]s_k + w_k \notag \\
                &\approx \nabla^2 f(x_k) \left [ \int_0^1 \myd\tau  \right
                  ]s_k + w_k \notag \\
                &= \nabla^2 f(x_k) s_k + w_k.
\end{align}
We can replace the unknown Hessian $\nabla^2 f(x_k)$ with the same GP
as used above, so that the measurement equation now becomes
\begin{align}
  {y}_k &= (s_k^\Transp \otimes I) h(x_k) + w_k.
\end{align}
Therefore, in order to compute the posterior distribution of
$h(x_k)$ given the most recent $p+1$ observations, we
follow an analogous path to the above discussion and employ~\eqref{eq:16}, but where the required terms are now instead given by
\begin{align}
  m_i = \Exp{{y}_i} = \bar{D}_i \mu(x_i),
\end{align}
and
\begin{align}
	\notag
  K_{i,j} &= \Exp{({y}_i-m_i)({y}_j-m_j)^\Transp}\\
  &= \bar{D}_i \kappa(x_i,x_j) \bar{D}_j^\Transp + R \delta_{i,j},
\end{align}
where $\delta_{i,j}$ is defined in~\eqref{eq:14delta}, and
\begin{align}
  K_{x, j} &= \Exp{ \left ( h(x) 
  - \mu(x)\right)\left({y}_j - m_j\right)^{\Transp}} 
 = \kappa(x, x_j)\, \bar{D}_j^{\Transp}.
	\label{eq:GP:Cov34}
\end{align}

\subsection{Computing the search direction}
In the deterministic optimisation setting, convergence to local minima
often requires that the search direction $p_k$ is a so-called descent
direction that satisfies (see e.g. \cite{NocedalW:2006})
\begin{align}
  \label{eq:12}
  p_k^\Transp \nabla f(x_k) < 0.
\end{align}
This essentially means that $p_k$ is sufficiently aligned with the negative
gradient direction $-\nabla f(x_k)$, which then affords local progress
along the direction $p_k$. 

When the search direction is determined by a scaling matrix
\begin{align}
  \label{eq:13}
  p_k = - B_k \nabla f(x_k),
\end{align}
then the descent condition can be guaranteed if $B_k \succ 0$ since
$-\nabla f^\Transp (x_k) B_k \nabla f(x_k) < 0$ for all
$\nabla f(x_k) \neq 0$. This is sufficient, but not necessary. If
$B_k$ is chosen as the inverse Hessian then there is no guarantee that
$B_k \succ 0$ unless the problem is strongly convex. In the more
general non-convex setting then the Hessian can be indefinite,
particularly in early iterations, and this should be carefully dealt
with. Interestingly, standard BFGS and DFP quasi-Newton methods
combat this by maintaining positive definite matrices $B_k$ as the
algorithm progresses (see e.g. \cite{NocedalW:2006}).

In the stochastic setting of this paper, we cannot guarantee a descent
direction since we do not have access to the gradient $\nabla
f(x_k)$, but only a noisy version ${g}_k$. In the current
section, we propose a mechanism to ensure that the search direction is
a descent direction in expectation. That is, we will compute a search
direction 
\begin{align}
  \label{eq:15}
  {p}_k &= -B_k {g}_k
\end{align}
based on the GP quasi-Newton approximatation from
Section~\ref{sec:LocalHessian} and show that
\begin{align}
  \label{eq:17}
  \Exp{{p}_k^\Transp \nabla f(x_k)\ \mid \ x_k, s_{\ell_{k-2}}, {y}_{\ell_{k-2}}} < 0.
\end{align}
In particular, let the matrix $B_k$ be defined as 
\begin{subequations}
\label{eq:24}
\begin{align}
  \label{eq:18}
  B_k = \left ( H_k + \lambda_k I \right )^{-1},
\end{align}
where $H_k$ comes from the mean value of the GP quasi-Newton
approximation based on $\{s_{\ell_{k-2}}, {y}_{\ell_{k-2}}\}$.  More
specifically, using $\phi_{\ell_{k-2}}(x)$ from \eqref{eq:19}, $H_k$
is defined as (note
that if $k<2$, then $\{s_{\ell_{k-2}}, {y}_{\ell_{k-2}}\}$ is empty
and $\phi_{\ell_{k-2}}(x) = \mu(x)$)
\begin{align}
  \label{eq:9}
  H_k = D \phi_{\ell_{k-2}}(x_k),
\end{align}
where $D$ is the duplication matrix defined in
\eqref{eq:halfVecOperator}, and $\lambda_k$ is selected as 
\begin{align}
  \label{eq:21}
  \lambda_k = \epsilon - \min \{ 0 , \eta_k\},
\end{align}
\end{subequations}
where $\eta_k$ is the minimum eigenvalue of $H_k$ and $\epsilon >0$ is
a user-defined tolerance. With this choice for the scaling matrix
$B_k$, we have the following Lemma.
\begin{lemma}
  Let the matrix $B_k$ be given by \eqref{eq:24} and the search
  direction $p_k$ be given by \eqref{eq:15}, then 
  \begin{align}
    \label{eq:25}
    \Exp{{p}_k^\Transp \nabla f(x_k)\ \mid \ x_k, s_{\ell_{k-2}}, {y}_{\ell_{k-2}}} < 0.
  \end{align}
\end{lemma}
\begin{proof}
Since we are conditioning on
$\{ x_k, s_{\ell_{k-2}}, y_{\ell_{k-2}}\}$, then $H_k$ is
deterministic. Furthermore, since $H_k$ is symmetric by construction and
$\lambda_k > 0$ is chosen such that ${H_k + \lambda_k I \succ 0}$, then
$B_k \succ 0$ is also deterministic. Then,
\begin{align*}
  &\Exp{p_k^\Transp \nabla f(x_k)} 
= \Exp{-(\nabla f(x_k) +
  v_k)^\Transp B_k \nabla f(x_k)},\\
  &= -\Exp{-v_k}^\Transp B_k \nabla
    f(x_k) - \nabla
    f^\Transp(x_k) B_k \nabla
    f(x_k),\\
  &< 0.
\end{align*}
\end{proof}


\section{Stochastic line search}
\label{sec:sls}
It is typical within the stochastic optimisation setting to ensure 
stability of the Markov-chain
\begin{align}
  \label{eq:7}
  x_{k+1} = x_k + \alpha_k p_k
\end{align}
by requiring that the step-length $\alpha_k$ satisfies
\begin{align}
  \label{eq:8}
  \sum_{k=1}^\infty \alpha_k = \infty, \qquad \sum_{k=1}^\infty \alpha_k^2 < \infty,
\end{align}
with a standard choice being $\alpha_k = \alpha_0 / k$ with $\alpha_0
> 0$. In addition to this, the
search direction $p_k$ must also satisfy certain conditions, that are
typically met when $p_k = -g_k$ for example (see
\cite{BottouCN:2018} for an excellent review of these conditions and
associated convergence proofs). 

While these choices are often sufficient, they are also known to be
conservative \cite{AsiD20129siam,MoulinesB:2011}. The initial
step-length $\alpha_0$ is often forced to be small in order to provide
stability, since the optimal choice depends on certain Lipschitz
constants \cite{BottouCN:2018}, that are typically not easy to obtain
for many practical problems. One way to combat this sensitivity is to
scale the gradient by the inverse Hessian, but this alone is not
sufficient to ensure stability. Another approach is to use adaptive
scaling methods that have found enormous impact in the machine
learning literature (see \cite{KingmaB:2015} and
\cite{luo2018adaptive} for example). A further mechanism, used in
deterministic optimisation, is to select $\alpha_k$ such that the new
cost $f(x_k + \alpha_k p_k)$ sufficiently decreases, which is often
called a line-search procedure.

In this section we will discuss a line-search procedure for stochastic
optimisation. In early iterations it mimics a deterministic
line-search using the noisy cost function to regulate the
step-length. As the iteration $k$ increases, we converge to taking
steps that mimic $\alpha_k = \alpha_0 / k$, which then affords
standard convergence results \cite{BottouCN:2018}. In this sense, the
line-search procedure can be interpreted as a mechanism to find
$\alpha_0$ such that steps $\alpha_k = \alpha_0 / k$ are likely to
produce a stable Markov-chain.

Towards this end, in the deterministic setting, the question of how far to
move in the search direction $p_k$ can be formulated as the following
scalar minimisation problem
\begin{align}
  \label{eq:10}
  \min_{\alpha} f(x_k + \alpha p_k), \quad \alpha > 0.
\end{align}
One commonly used approach is to aim for a sub-optimal solution
to~\eqref{eq:10} that guarantees a sufficient decrease in the
cost. For example, the popular Armijo condition \cite{Armijo:1966}
requires a step-length~$\alpha_k$ such that 
\begin{align}
  \label{eq:DetWolfe}
  f(x_k + \alpha p_k) \leq f(x_k) + c \alpha_k g_k^{\Transp} p_k,
\end{align}
where the user-defined constant $c\in (0, 1)$. The above condition is
also known as the first Wolfe condition \cite{Wolfe:1969, Wolfe:1971}.

In the stochastic setting, we unfortunately do not have access to
either $f(x)$ or $\nabla f(x)$ but only noisy versions of them. This
presents the following difficulties in employing an Armijo type condition:
\begin{itemize}
\item We may accept a step-length $\alpha_k$ in the case where the
  observed cost has sufficiently decreased, even though the true cost
  may in fact have increased;
\item We may reject a suitable $\alpha_k$ when the observed cost
  increased, even though the true cost may have decreased sufficiently.
\end{itemize}

Consider the case when the measurements of the function and its
gradient are given by
\begin{align}
  \label{eq:CostGradMeas}
  \widehat{f}(x_k)=f(x_k)+e_k, \qquad {g}_k= \nabla f(x_k) + v_k,
\end{align}
where $e_k$ and $v_k$ denote independent noise on the function and
gradient evaluations, respectively. Furthermore we assume that
  \begin{alignat}{6}
    \Exp{e_k}&=b, \quad& &\Cov{e_k} = \sigma_f^2, \quad
    \Exp{v_k}&=0, 
    \quad& &\Cov{v_k} = R.
  \end{alignat}
Since $\widehat{f}$ and ${g}_k$ are random variables, we explore the idea of
requiring~\eqref{eq:DetWolfe} to be fulfilled in expectation when the
exact quantities are replaced with their stochastic counterparts, that is
\begin{align}
  \label{eq:StochWolfe1}
  \Exp{\widehat{f}(x_k + \alpha {p}_k) - \widehat{f}(x_k) -
  c\alpha_k {g}_k^{\Transp} {p}_k \mid x_k,
  s_{\ell_{k-2}}, \widehat{y}_{\ell_{k-2}}} \leq 0,
\end{align}
This is certainly one way in which we can reason about the Armijo
condition in the stochastic setting we are interested in. Although
satisfaction of \eqref{eq:StochWolfe1} does not leave any guarantees
when considering a single measurement, it still serves as an important
property that could be exploited to provide robustness for the entire
optimisation procedure. To motivate our proposed algorithm we hence
start by establishing the following results.
\begin{theorem}[Stochastic Armijo Condition]
	\label{lem:sW1}
	Assume that 
	\begin{itemize}
        \item[]A1: $\quad f(\cdot)$ is twice continuously
          differentiable on an open set $\mathcal{X} \subseteq \R^{n_x}$;
        \item[]A2: the gradient is unbiased, $\Exp{{g}_k}=\nabla f(x_k)$;
        \item[]A3: the expected cost obeys $\Exp{\widehat{f}(x_k)}=f(x_k)+b$;
        \item[]A4: a descent direction is expected, $\Exp{{p}_k^\Transp \nabla f(x_k)}<0$. 
	\end{itemize}
	Then (for $\alpha$ small)
          \begin{align}
            \label{eq:StochWolfe1a}
            \Exp{\widehat{f}(x_k + \alpha {p}_k) - \widehat{f}(x_k) -
            c\alpha_k {g}_k^{\Transp} {p}_k \mid x_k,
            s_{\ell_{k-2}}, \widehat{y}_{\ell_{k-2}}} \leq 0,
          \end{align}
	where 
	\begin{align}
          \label{eq:SLSsW1c1}
          0<c<\bar{c} = \frac{\nabla f^\Transp(x_k) B_k \nabla
          f(x_k)}{\nabla f^\Transp(x_k) B_k \nabla f(x_k) + \Tr \{B_kR\}}.
	\end{align} 
\end{theorem}
\begin{proof}
  See Appendix~\ref{app:StochWolfeCond1}.
\end{proof}
Note that as the gradient noise variance $R$ vanishes, we recover
the classical result that $\bar{c} < 1$ (see
e.g. \cite{NocedalW:2006}). 

%
Relying upon this result, we propose a line search with psuedo-code
given in Algorithm \ref{alg:SBLS}. An input to this algorithm is the
search direction ${p}_k$, which can be computed using any preferred
method.  The step length is initially set to be the minimum of the
natural step length 1 and the iteration dependent value $\xi/k$.  In
this way the initial step length is kept at 1 until $k>\xi$, a point
after which it is decreased at the rate $1/k$.  

Then we check whether the new point $x_k+\alpha {p}_k$ satisfies the stochastic Armijo condition.  If this is not the case, we decrease~$\alpha_k$
with the scale factor~$\rho$. This is repeated until the condition is
met, unless we hit an upper bound $\max\{0,\tau-k\}$ on the number of
backtracking iterations, where $\tau>k$ is a positive integer.  With
this restriction the decrease of the step length is limited, and when
$k\ge\tau$ we use a prescribed formula for step length no matter if
the stochastic Armijo condition is satisfied or not.
It should be noted that for finite $\xi >0$ and finite $\tau > 0$, the
following algorithm will eventually take steps according to
$\alpha_k = \frac{\xi}{k}$
and thus mimic the step-length properties required for convergence
in standard stochastic optimisation algorithms (see \cite{BottouCN:2018}).

\begin{algorithm}[htb!]
	\caption{\textsf{Stochastic backtracking line search}}
	\begin{algorithmic}[1]
          \REQUIRE Iteration index $k$, spatial point $x_k$, search
          direction ${p}_k$, scale factor $\rho\in (0,1)$, reduction
          limit $\xi\ge1$, backtracking limit
          $\tau>0$. 
		\STATE Set the initial step length $\alpha_{k}=\min\{1,\xi/k\}$
		\STATE Set $i=1$
		\WHILE{$\widehat{f}(x_k + \alpha {p}_k) > \widehat{f}(x_k) + c\alpha_k {g}_k^{\Transp} {p}_k$ \textbf{and} $i\le\max\{0,\tau-k\}$}
			\STATE Reduce the step length $\alpha_{k} \leftarrow \rho \alpha_k$
			\STATE Set $i\leftarrow i+1$
		\ENDWHILE
	\end{algorithmic}
	\label{alg:SBLS}
\end{algorithm}


\section{Resulting algorithm}
\label{sec:ResAlg}
In this brief section, we summarise the main algorithm employed in the
simulations to follow. 
\begin{algorithm}[htb!]
  \caption{\textsf{Stochastic quasi-Newton}}
  \begin{algorithmic}[1]
    \REQUIRE \textbf{Search direction}: a mean and covariance function
    $\mu$ and $\kappa$, respectively, for the prior GP according to
    \eqref{eq:GPprior4Hessian}. A covariance $R \succ 0$ for the
    gradient noise $v_k$ and a memory length $p >
    0$. \textbf{Step-length}: parameters according to
    Algorithm~\ref{alg:SBLS}. A maximum number of iterations
    $k_{\max} > 0$ and set $k=0$ and and $x_0 \in \mathcal{X}$.
    \WHILE{$k < k_{\max}$} 
    \STATE Obtain the gradient estimate $g_k$.
    \IF{$k > 0$} 
    \STATE Compute $y_{k-1}$ and $s_{k-1}$.
    \ENDIF
    \STATE Compute the search direction $p_k = -B_kg_k$, where $B_k$
    is determined by \eqref{eq:18}.
    \STATE Calculate $\alpha_k$ using Algorithm~\ref{alg:SBLS} and $p_k$.
    \STATE Compute $x_{k+1} = x_k + \alpha_k p_k$.
    \STATE Update $k \leftarrow k+1$.
    \ENDWHILE
  \end{algorithmic}
  \label{alg:SQN}
\end{algorithm}


\section{Application to nonlinear system identification}
\label{sec:AppNLSYSID}
The method developed above is indeed generally applicable to a broad
class of stochastic and non-convex optimization problems. By way of
illustration we will consider its application to the problem of
maximum likelihood identification of nonlinear state space models
\begin{subequations}
\begin{align}
x_{t} &= f(x_{t-1},\theta) + w_t,\\
y_t &= g(x_t, \theta) + e_t,
\end{align}
\end{subequations}
where $x_t$ denotes the state, $y_t$ denotes the measurements and
$\theta$ denotes the unknown (static) parameters. The two nonlinear
functions $f(\cdot)$ and $g(\cdot)$ denotes the nonlinear functions
describing the dynamics and the measurements,
respectively. Furthermore, the process noise is Gaussian distributed
with zero mean and covariance $Q$, $w_t\sim \N(0, Q)$ and the
measurement noise is given by $e_t\sim\N(0, R)$. Finally, the initial
state is distributed according to $x_0 \sim p(x_0\mid\theta)$. The
problem we are interested in is to estimate the unknown
parameters~$\theta$ by making use of the available measurements
$y_{1:N} = \{y_1, y_2, \dots, y_{N}\}$ to maximize the likelihood
function $p(y_{1:N} \mid \theta)$
\begin{align}
\label{eq:App:MaxLik}
\widehat{\theta} = \argmax{\theta}{p(y_{1:N} \mid \theta)}.
\end{align}
The likelihood function can via repeated use of conditional
probabilities be rewritten as
\begin{align}
p(y_{1:N}\mid \theta) = \prod_{t=1}^{N} p(y_t\mid y_{1:t-1},\theta),
\end{align}
with the convention that $y_{1:0} = \emptyset$. The one step ahead
predictors are available via marginalization
\begin{align}
\notag
p(y_t&\mid y_{1:t-1}, \theta) = \int p(y_t, x_t \mid y_{1:t-1}, \theta) \myd x_t\\
&= \int p(y_t\mid x_t, \theta)  p(x_t\mid y_{1:t-1}, \theta) \myd x_t.
\end{align}
One intuitive interpretation of the above integral is that it
corresponds to averaging over all possible values for the state
$x_{t}$. The challenge is of course how to actually compute this
integral. By making use of particle filter 
\cite{Gordon:1993,Kitagawa:1993,StewartM:1992} to approximate the
likelihood we are guaranteed to obtain an unbiased
estimate~\cite{Delmoral:2004,MalikP:2011}. Likelihood gradients can
also be calculated using particle methods~\cite{PoyiadjisDS:2011},
which we employ in the simulations below.


The particle filter---which is one member of the family of sequential
Monte Carlo (SMC) methods---has a fairly rich history when it comes to
solving nonlinear system identification problems. For introductory
overviews we refer to~\cite{SchonLDWNSD:2015,Kantas:2015}.

\subsection{Numerical example -- scalar linear SSM}
The first example to be considered is the following simple linear time
series
\begin{subequations}
\begin{equation}
  \label{eq:linsys}
  \begin{array}{r}
    x_{t+1} = ax_t + v_t,\\
    y_t = cx_t + e_t,\\
  \end{array} \quad
  \begin{bmatrix}
    v_t \\ 
    e_t
  \end{bmatrix}
  \sim \mathcal{N} \left (
    \begin{bmatrix}
      0 \\ 
      0
    \end{bmatrix}
    ,
    \begin{bmatrix}
      q & 0 \\ 
      0 & r
    \end{bmatrix}
  \right )
\end{equation}
with the true parameters given by
$\theta^\star = \left [ a^\star, c^\star, q^\star, r^\star \right ] =
  \left [ 0.9, 1.0, 0.1, 0.5 \right ].$
\end{subequations}
Here we wish to estimate all four parameters $a,c,q$ and $r$ based on
observations $y_{1:N}$. The reason for including this example is that this
problem is generally considered to be solved. Indeed, using any
standard quasi-Newton line-search algorithm in combination with a
Kalman filter to evaluate the log-likelihood cost and
gradient vector will provide the necessary components to estimate
$\theta$. 

This is not a stochastic optimisation problem, but it is included so
that our new GP based quasi-Newton Algorithm (henceforth referred to
as QN-GP) can be profiled in a situation where standard algorithms
also apply. Therefore, in order to compare these methods, we generated
$100$ data sets $y_{1:N}$ with $N=100$ observations in each. The standard
quasi-Newton line-search algorithm was run as usual with no
modification. However, when running the new QNGP algorithm (and only
in this case) an extra noise term was added to each and every log-likelihood and gradient
evaluation. In particular, for the QNGP case we generated noisy
log-likelihoods and gradients via (where the noise realisation was
changed for every evaluation)
\begin{subequations}
\begin{align}
  \widehat{\ell}(\theta) &\triangleq \ell(\theta) + \eta, \quad   &\eta &\sim \mathcal{N}(0,1), \\
  \widehat{\nabla_\theta \ell}(\theta) &\triangleq \nabla_\theta
  \ell(\theta) + \eta_g, 
  &\eta_g &\sim \mathcal{N}(0,I).
\end{align}
\end{subequations}
Figure~\ref{fig:bode_compare} shows the Bode response for each Monte
Carlo run and for each Algorithm. The vast majority of estimates are
grouped around the true Bode response, with one or two estimates
being trapped in local minima that results in a poor
estimate. Interestingly, this is true for both the standard
quasi-Newton algorithm and the new QNGP algorithm. Both algorithms
were initialised with the same parameter value $\theta_0$, which
itself was obtained by generating a random stable system. 

\begin{figure}[!bth]
    \centering
    \begin{subfigure}[t]{\linewidth}
        \includegraphics[width=\columnwidth]{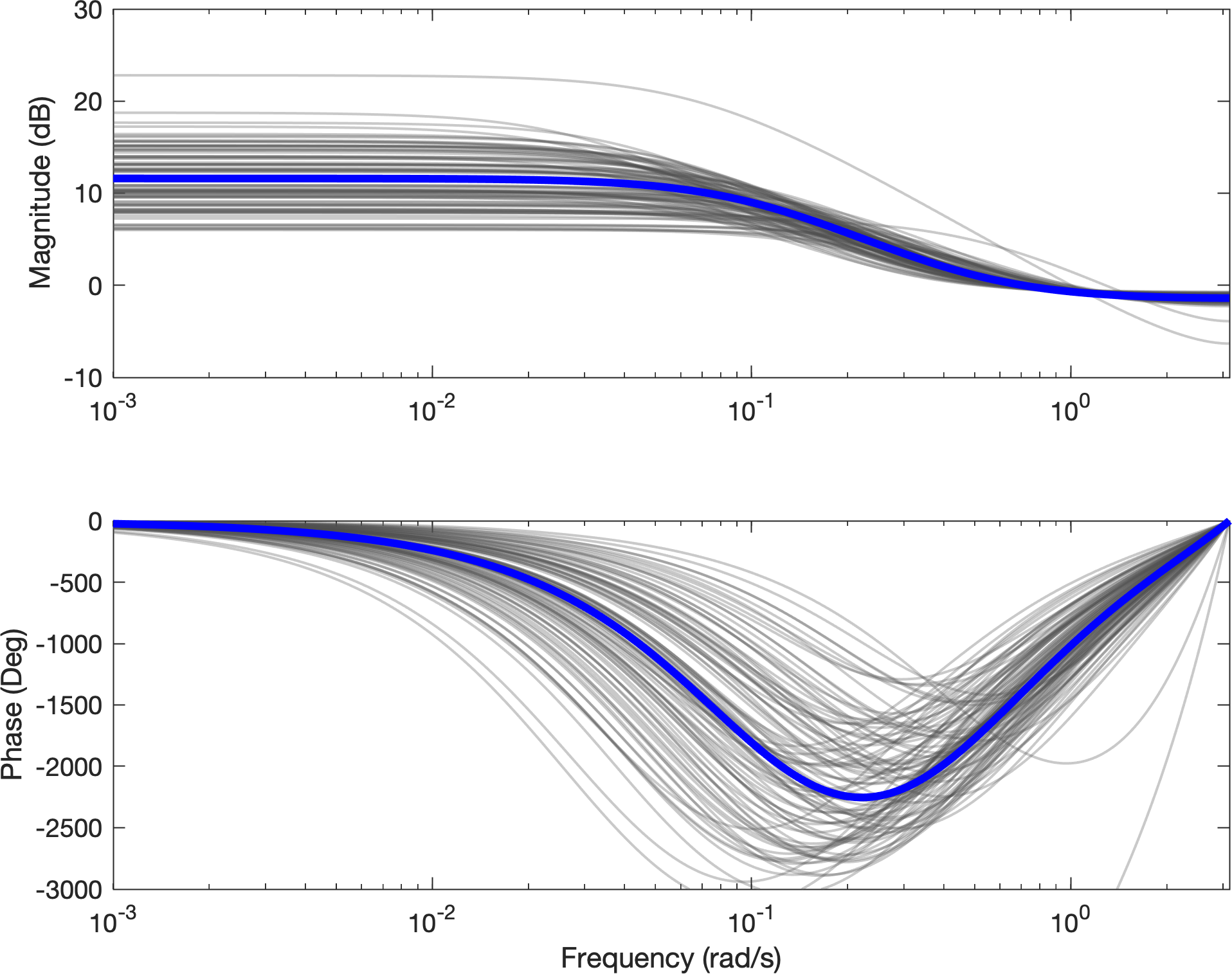}
       \caption{Standard quasi-Newton algorithm with no noise on
         function values or gradients.}
        \label{fig:lss_gn}
      \end{subfigure}
      \begin{subfigure}[t]{\linewidth}
        \includegraphics[width=\columnwidth]{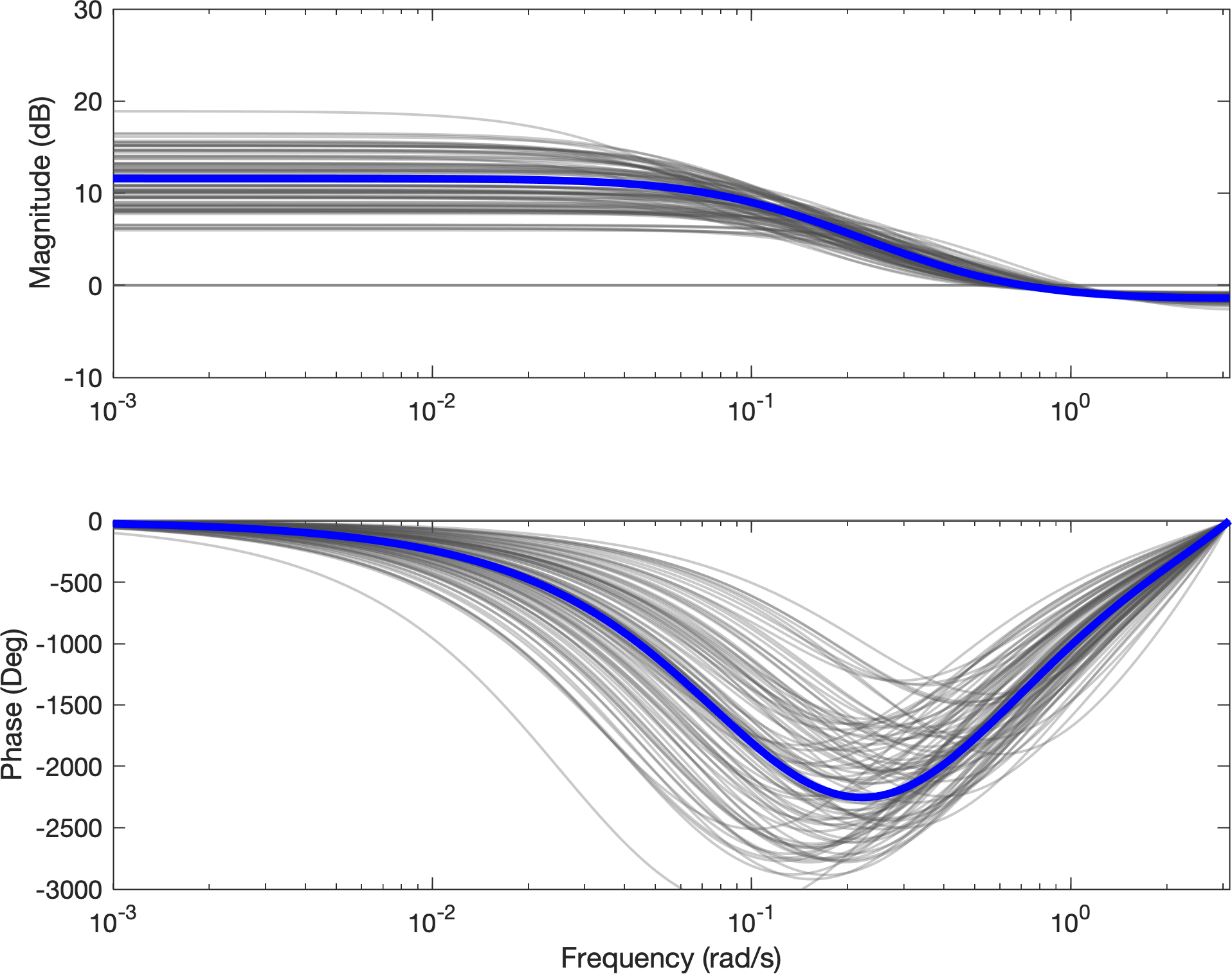}
       \caption{Quasi-Newton Gaussian process algorithm.}
        \label{fig:lss_qngp}
      \end{subfigure}
      \caption{Bode plots of estimated mean (light grey) and true
        (blue) systems for $100$ Monte Carlo runs.}
      \label{fig:bode_compare}
\end{figure}

\subsection{Numerical example -- nonlinear toy problem}
A commonly employed nonlinear benchmark problem involves the
following system
\begin{subequations}
  \label{eq:nlp}
  \begin{align}
    x_{t+1} &= ax_t + b\frac{x_t}{1 + x_t^2} + c\cos (1.2t) + v_t, \\
    y_t &= dx_t^2 + e_t,\\
  \begin{bmatrix}
    v_t \\ 
    e_t
  \end{bmatrix}
  &\sim \mathcal{N} \left (
    \begin{bmatrix}
      0 \\ 
      0
    \end{bmatrix}
    ,
    \begin{bmatrix}
      q & 0 \\ 
      0 & r
    \end{bmatrix}
  \right )
  \end{align}
\end{subequations}
where the true parameters are 
$\theta^\star = \left [ a^\star, b^\star, c^\star, d^\star, q^\star,
    r^\star \right ] = \left [ 0.5, 25, 8, 0.05, 0, 0.1 \right ].$
Here we repeat the simulation experiment from \cite{SchonWN:2011},
where a Monte Carlo study was performed using $100$
different data realisations $Y_N$ of length $N=100$.  For each of
these cases, an estimate $\widehat\theta$ was computed using $1000$
iterations of Algorithm~\ref{alg:SQN}. The algorithm was initialised
with the $i$'th element of $\theta_0$ chosen via $\theta_0(i) \sim
\mathcal{U}\left ( \frac{1}{2} \theta^\star(i),\ \frac{3}{2}
  \theta^\star(i) \right )$. In all cases
$M=50$ particles were used.

Figure~\ref{fig:nlss_toy} shows the parameter iterates (for
$a,b,c,d$). This shows all Monte Carlo runs (note that non were
trapped in a local minima). By way of comparison, the method presented
in this paper is compared with the EM approach from
\cite{SchonWN:2011} and the results are provided in
Table~\ref{tab:ex3}, where the values are the sample mean of the
parameter estimate across the Monte Carlo trials plus/minus the sample
standard deviation. For the EM approach, 8/100 simulations were
trapped in minima that were far from the global minimum and these
results have been removed from the calculations in Table~\ref{tab:ex3}.

\begin{figure}[!bth]
    \centering
    \includegraphics[width=\columnwidth]{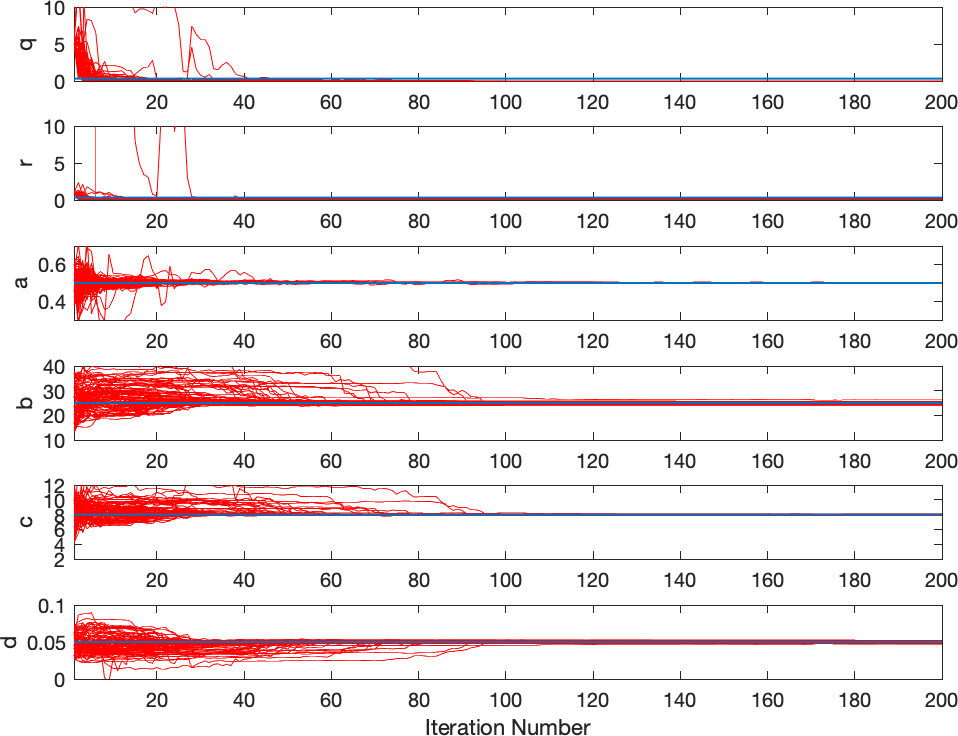}
    \caption{Parameter iterations for nonlinear benchmark
      problem. True value (solid blue line), parameter evolution
      (think red line per simulation).}
    \label{fig:nlss_toy}
\end{figure}

\begin{table}[htb]
  \centering 
  \caption{\em True and estimated parameter values for QNGP and PSEM algorithms; 
    mean value and standard deviations are shown for the estimates 
    based on $100$ Monte Carlo runs.}
  {\tiny
  \begin{tabular}{llll}
  	\toprule
    $\theta$ & $\theta^\star$ & \textbf{QNGP} & \textbf{PSEM}   \\
	\midrule
    $a$ & $0.5$  & $0.50 \pm 0.0011$  &  $0.50 \pm 0.0019$ \\ 
    $b$ & $25.0$ & $25.0 \pm 0.25$    & $25.0 \pm 0.99$  \\ 
    $c$ & $8.0$  & $8.0  \pm 0.03$    &  $7.99 \pm  0.13$ \\ 
    $d$ & $0.05$ & $0.05 \pm 0.0006$    & $0.05 \pm 0.0026$  \\
    $q$ & $0$    & $1.5 \times 10^{-6} \pm 1.1 \times 10^{-5}$    & $7.78 \times 10^{-5} \pm 7.6 \times 10^{-5}$  \\
    $r$ & $0.1$  & $0.095 \pm 0.015$    & $0.106 \pm 0.015$  \\
	\bottomrule
  \end{tabular}
}
  \label{tab:ex3}
\end{table}

\subsection{Numerical example -- interferometry}
Interferometry was recently made famous due to its use in Nobel award
winning detection of gravitational waves in 
2015~\cite{abbott2016observation}. Two 4km interferometers form
the Laser Interferometer Gravitational-Wave Observatory, called LIGO
for short. The measured intensity of each laser beam at time $t$, denoted here
as $y_{t,1}$ and $y_{t,2}$, can be plotted against one another to
reveal an elliptical shape as illustrated in
Figure~\ref{fig:nlss_inter}. This shape is essential to detecting
reflector perturbations, but is typically not known. The aim is to
estimate this ellipse based on observed data. We can parametrize this relationship as
\begin{subequations}
\begin{align}
  y_{t,1} &= \alpha_0 + \alpha_1\cos(\kappa p_t) + e_{t,1}, \\
  y_{t,2} &= \beta_0 + \beta_1 \sin(\kappa p_t + \gamma) + e_{t,2}, 
\end{align}
\end{subequations}
where $e_{t,1} \sim \N(0, \sigma^2), e_{t,2} \sim \N(0, \sigma^2)$ and $p_t$ denotes the position of the laser reflector, which is not directly
measured. A suitable model for the dynamics is provided by a
simple discrete-time kinematic relationship
\begin{align}
  \begin{pmatrix}
    p_{t+1}\\
    v_{t+1}
  \end{pmatrix} = \begin{pmatrix}
    1 & \Delta\\
    0 & 1
  \end{pmatrix} \begin{pmatrix}
    p_t\\
    v_t
  \end{pmatrix} + w_t,
\end{align}
where $\Delta>0$ is the sample interval and $w_t$ is an unknown action
on the system. The unknown parameters are collected as
$\theta = \begin{pmatrix} {\alpha_0} & {\alpha_1} & {\beta_0} &
  {\beta_1} & \gamma & \sigma \end{pmatrix}^{\Transp}$ and the
resulting maximum likelihood system identification problem can be
expressed as~\eqref{eq:App:MaxLik}. 
This problem construction is known as a blind Wiener
problem~\cite{WillsSLN:2013}. We generated $N=1000$ samples of the
outputs using the above model with true parameter values indicated in
Table~\ref{tab:ex4}. To calculate the likelihood and its gradient, we
employed a particle filter with $M=50$ particles. A Monte-Carlo
simulation with 100 runs was performed and the results are provided in
Table~\ref{tab:ex4}. Figure~\ref{fig:nlss_inter} also provides the
measured outputs for one simulation and plots the ellipse as the
algorithm progresses. The final ellipse is indicated in solid black
while the true ellipse is shown in red. Based on these results, it
appears that Algorithm~\ref{alg:SQN} appears to be performing quite
well .

\begin{figure}[!bth]
    \centering
    \includegraphics[width=\columnwidth]{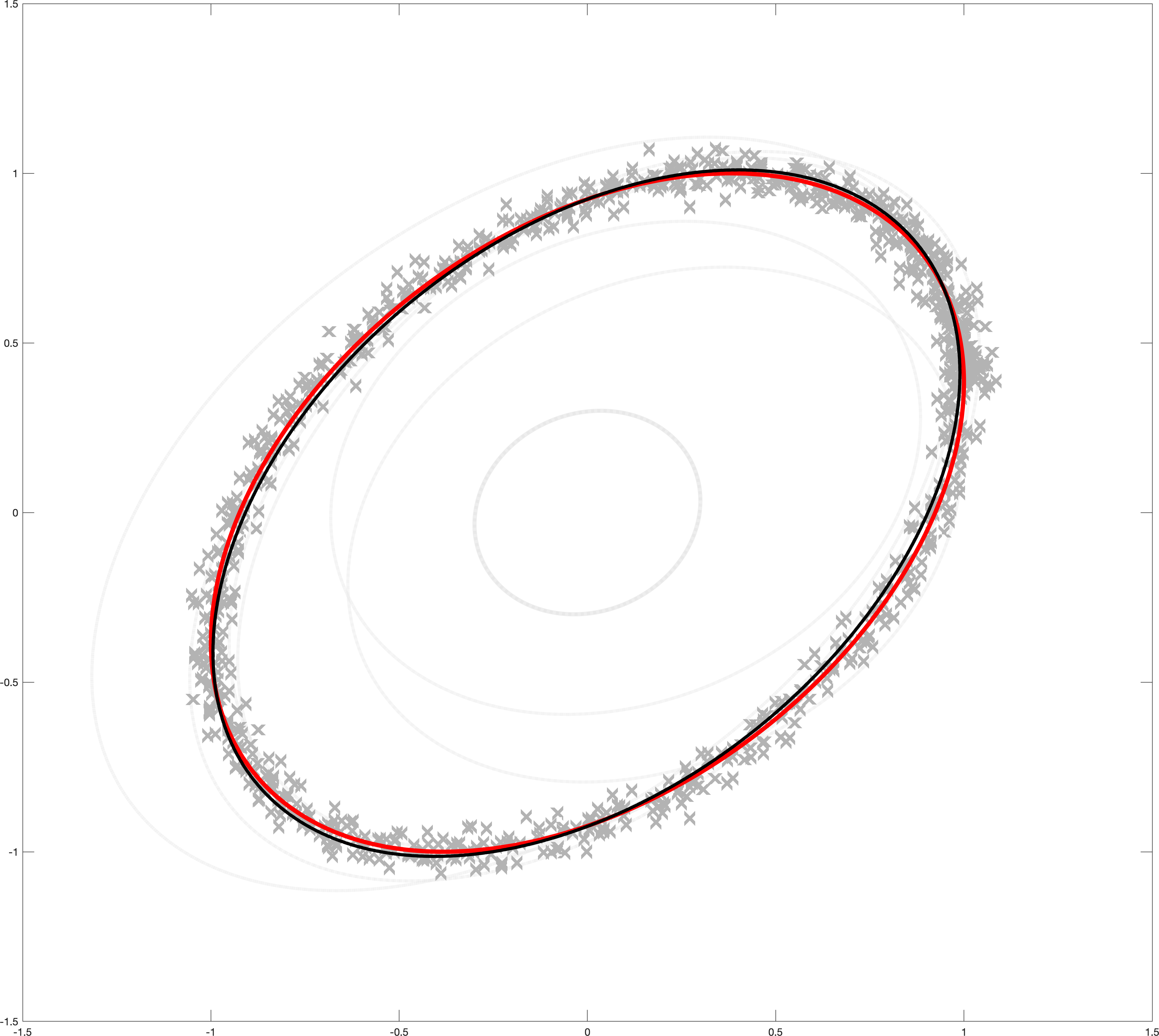}
    \caption{Parameter iterations for nonlinear interferometry
      problem. True ellipse (solid red), data points (crosses),
      intermediate ellipse estimates (light grey), final ellipse
      estimate (solid black).}
    \label{fig:nlss_inter}
\end{figure}

\begin{table}[htb]
  \centering 
  \caption{\em True and estimated parameter values for QNGP algorithm
    on the interferometry problem; mean value and standard deviations are shown for the estimates 
    based on $100$ Monte Carlo runs.}
  {\tiny
  \begin{tabular}{lll}
	\toprule
    $\theta$ & $\theta^\star$ & \textbf{QNGP Estimate} \\ 
	\midrule
    $\alpha_0$ & $0$  & $-0.0035 \pm 0.0026$ \\ 
    $\alpha_1$ & $1$ & $0.992 \pm 0.0034$    \\ 
    $\beta_0$ & $0$  & $-0.0001  \pm 0.0027$   \\ 
    $\beta_0$ & $1$ & $1.01 \pm 0.0042$   \\ 
    $\gamma$ & $0.4$    & $0.42 \pm 0.01$ \\ 
    $\sigma$ & $0.18$  & $0.14 \pm 0.0052$ \\ 
	\bottomrule
  \end{tabular}
}
\label{tab:ex4}
\end{table}

\subsection{MIMO Hammerstein-Wiener system}
\label{sec:Ex5}
As a further example, we turn attention to the 
multiple-input/multiple-output (MIMO) Hammerstein-Wiener system from
\cite{WillsSLN:2013} (Section 6.2 in that paper). The system
has two inputs, two outputs and 4'th order linear dynamics with
saturation and deadzone static nonlinearities on the input and output
channels. Full details of this system and the associated parametrization
can be obtained from \cite{WillsSLN:2013} and are not repeated here
due to page limitations.

For the purposes of estimation, $N=2000$ samples of the inputs and
outputs were simulated. In this case, two different algorithms
are compared:
\begin{enumerate}
\item The SMC based Expectation-Maximisation method developed in
  \cite{WillsSLN:2013}, called the PSEM (particle smoother EM) approach;
\item The quasi-Newton GP based solution presented in the current
  paper, denoted as QNGP.
\end{enumerate}
For QNGP method, $M=100$ particles were used, and for QNGP $M=2500$
particles were used (note that the computational complexity of the
former method scales as $O(M^2)$, while the second SMC approach scales
as $O(M)$, making the two comparable). The algorithms were terminated
after $1000$ iterations. The results of $100$ Monte Carlo runs for all
algorithms are shown in
Figures~\ref{fig:nlu1}--\ref{fig:bodemlem}. For each run, different
noise realisations were used according to the distributions specified
above. In each plot, the solid blue line indicates the true response,
while the red lines indicate the mean (think dashed) and one standard
deviation from mean (thin dashed).

This example shows that the proposed Algorithm~\ref{alg:SQN} compares
well to state-of-the-art methods such as PSEM. A potential benefit of
the QNGP approach is that only forward filtering is required, whereas
PSEM requires a SMC-based smoother, which can be challenging in
general\cite{LindstenS:2013,Kantas:2015}.

\begin{figure}[bth]
  \centering
  \begin{subfigure}{0.45\columnwidth}
    \includegraphics[width=\linewidth]{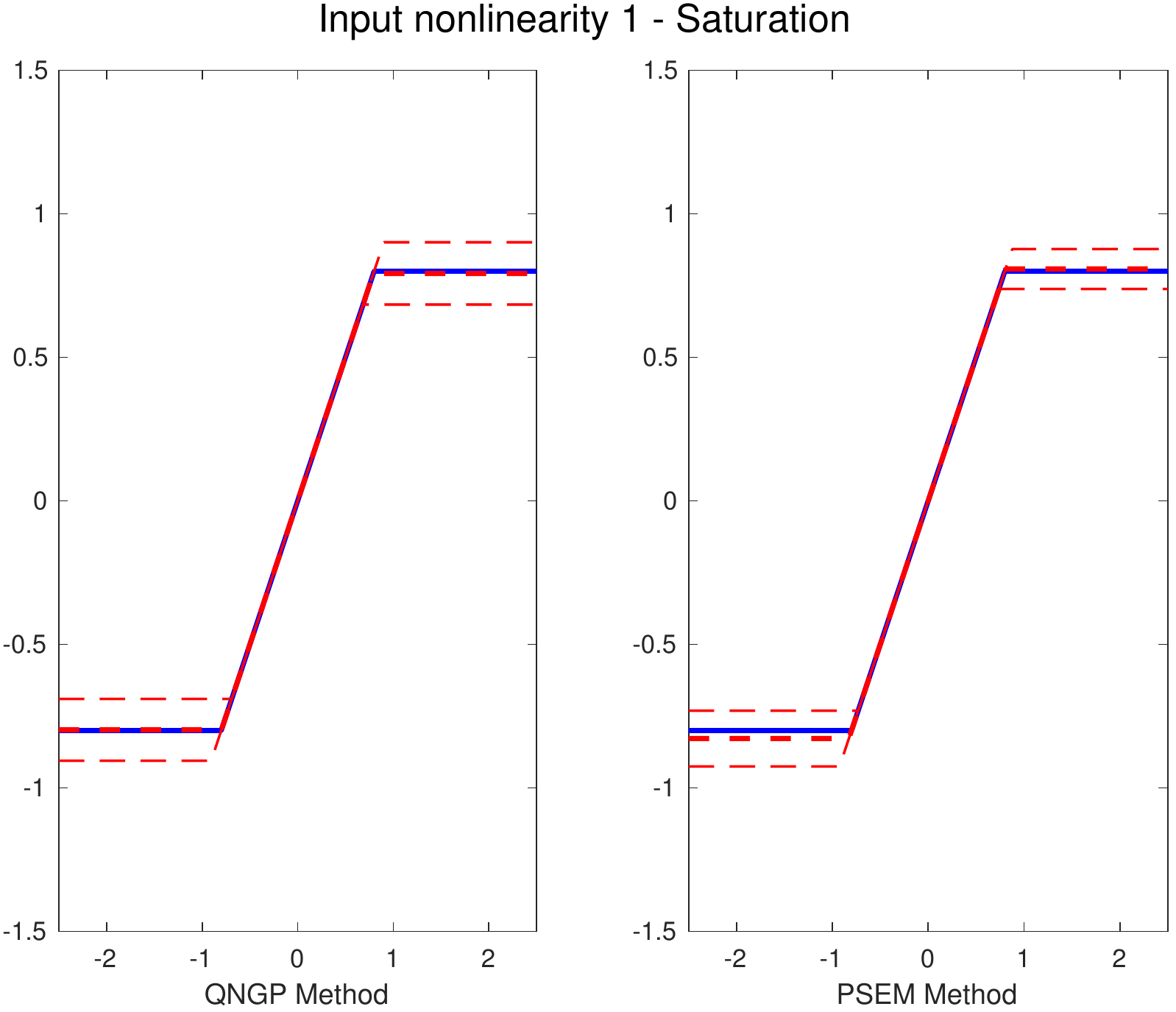}%
  \caption{Input nonlinearity-1.}
  \label{fig:nlu1}
  \end{subfigure}
  ~
  \begin{subfigure}{0.45\columnwidth }
    \includegraphics[width=\linewidth]{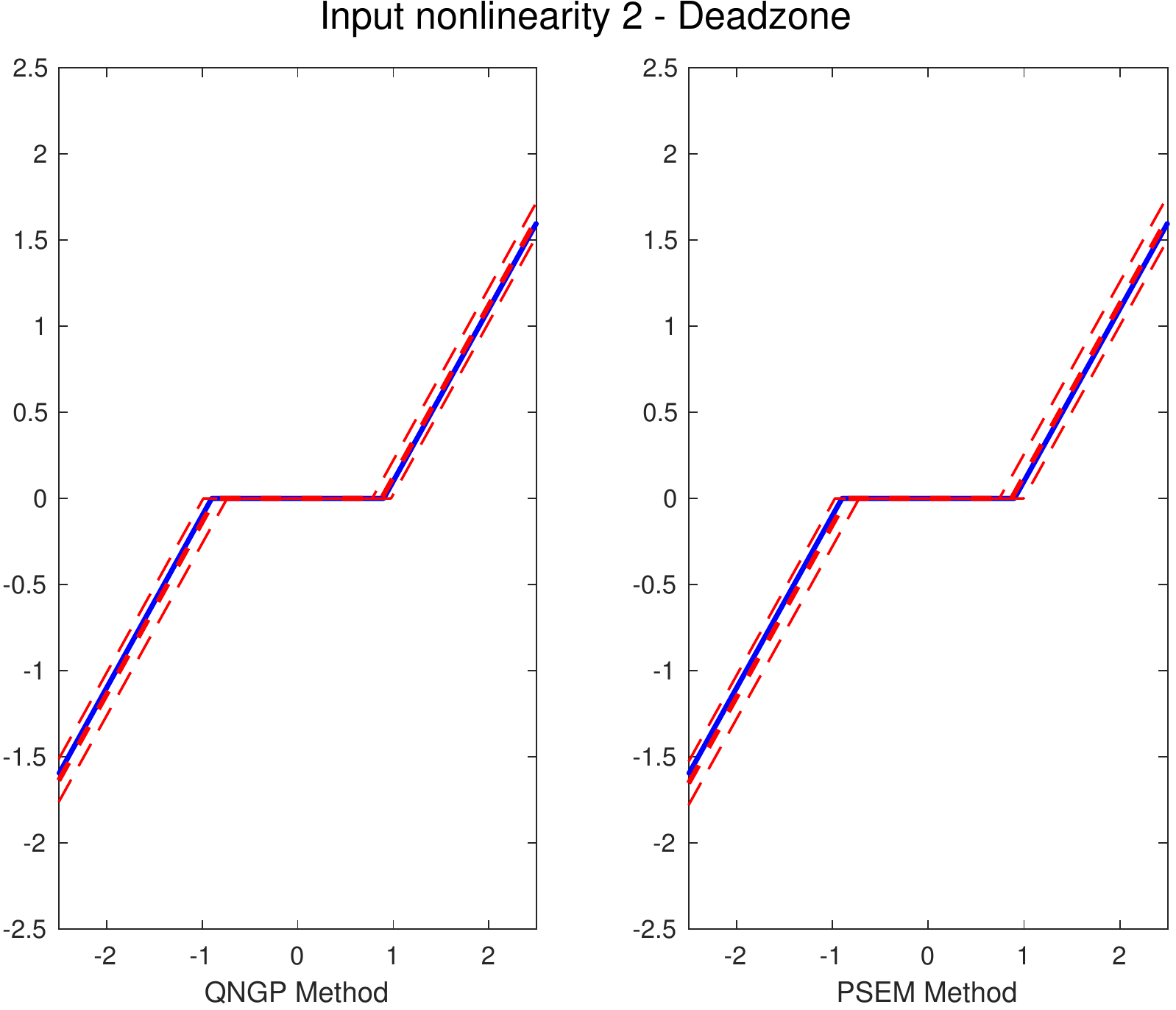}%
    \caption{Input nonlinearity-2.}
  \label{fig:nlu2}
  \end{subfigure}
  ~
  \begin{subfigure}{0.45\columnwidth}
  \includegraphics[width=\linewidth]{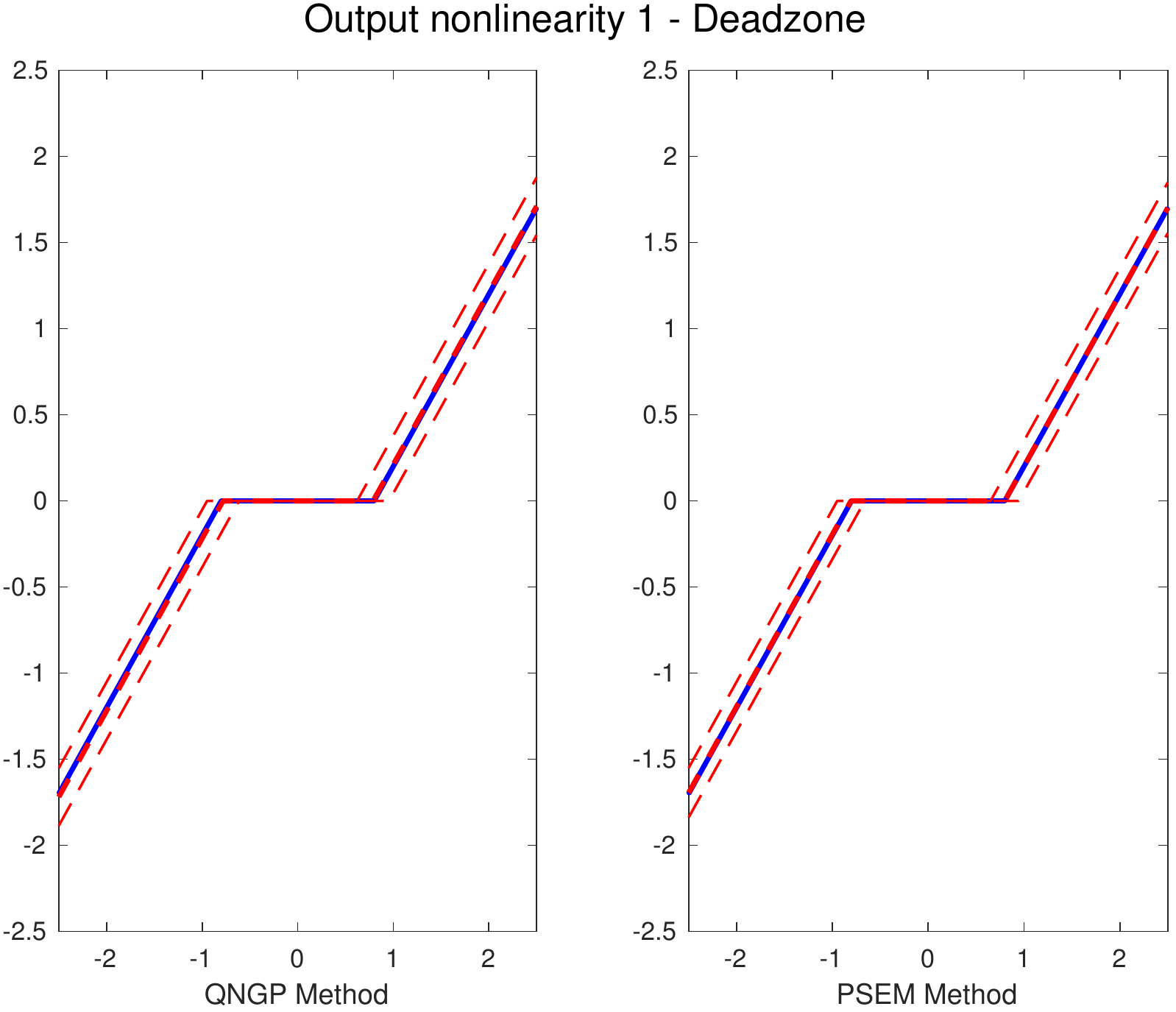}%
  \caption{Output nonlinearity-1.}
  \label{fig:nly1}
  \end{subfigure}
  ~
  \begin{subfigure}{0.45\columnwidth}
  \includegraphics[width=\linewidth]{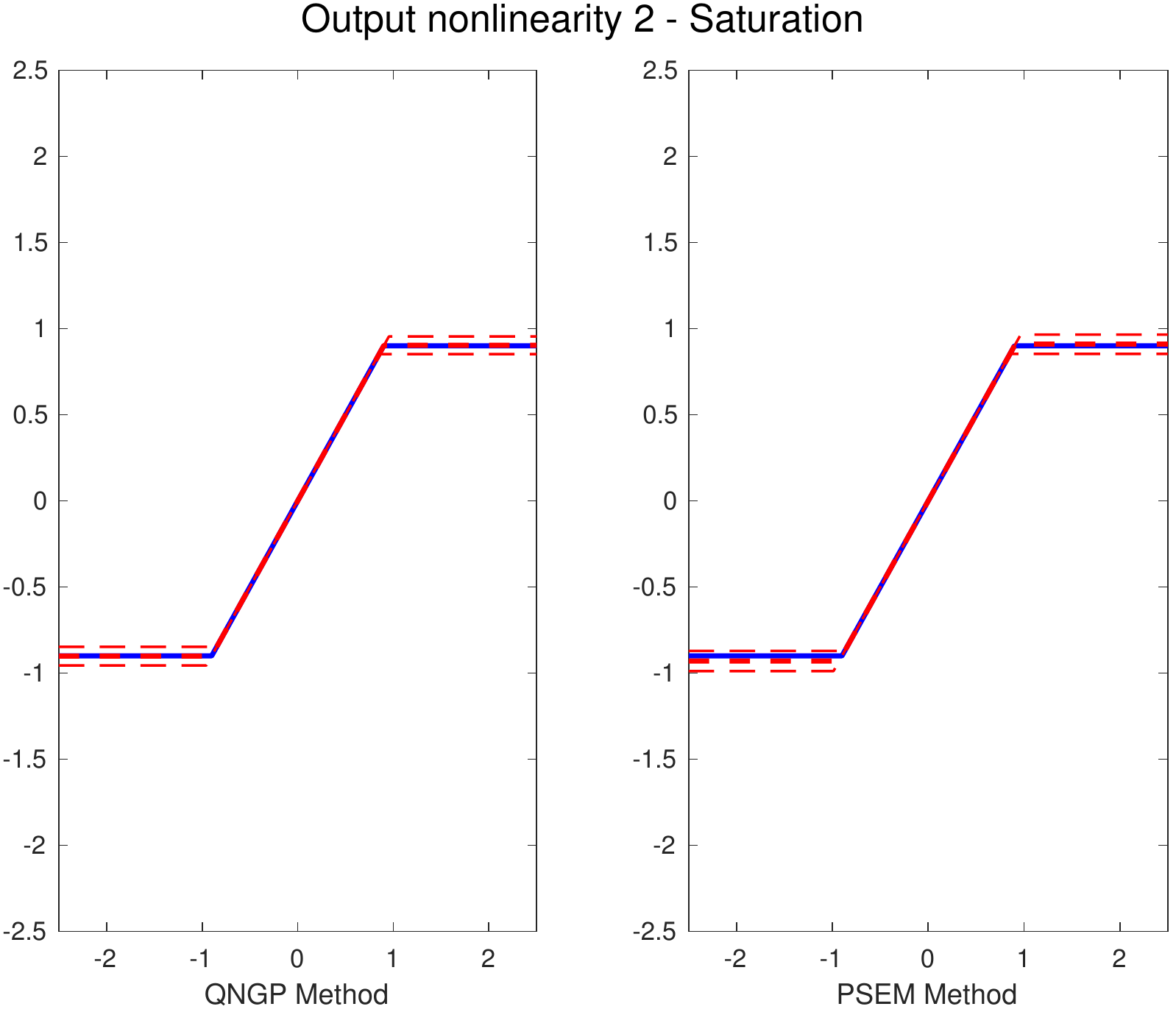}%
  \caption{Output nonlinearity-2.}
  \label{fig:nly2}
  \end{subfigure}
\end{figure} 




\begin{figure}[bth]
\begin{center}
  \includegraphics[width=0.95\columnwidth]{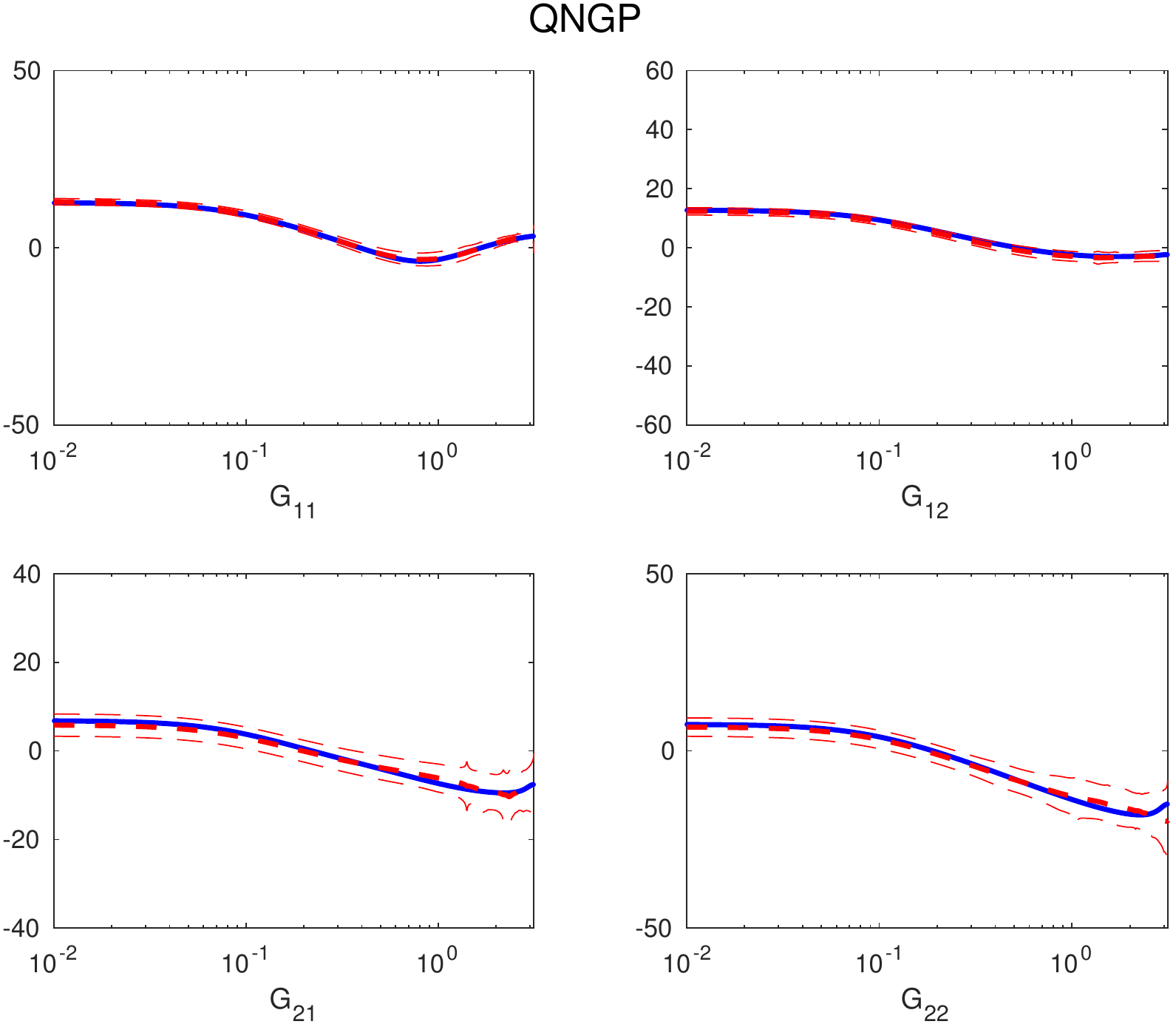}%
  \caption{Bode magnitude response using the QNGP method for
    the example studied in Section~\ref{sec:Ex5}. }
  \label{fig:bodeoe}
\end{center}
\end{figure} 

\begin{figure}[bth]
\begin{center}
  \includegraphics[width=0.95\columnwidth]{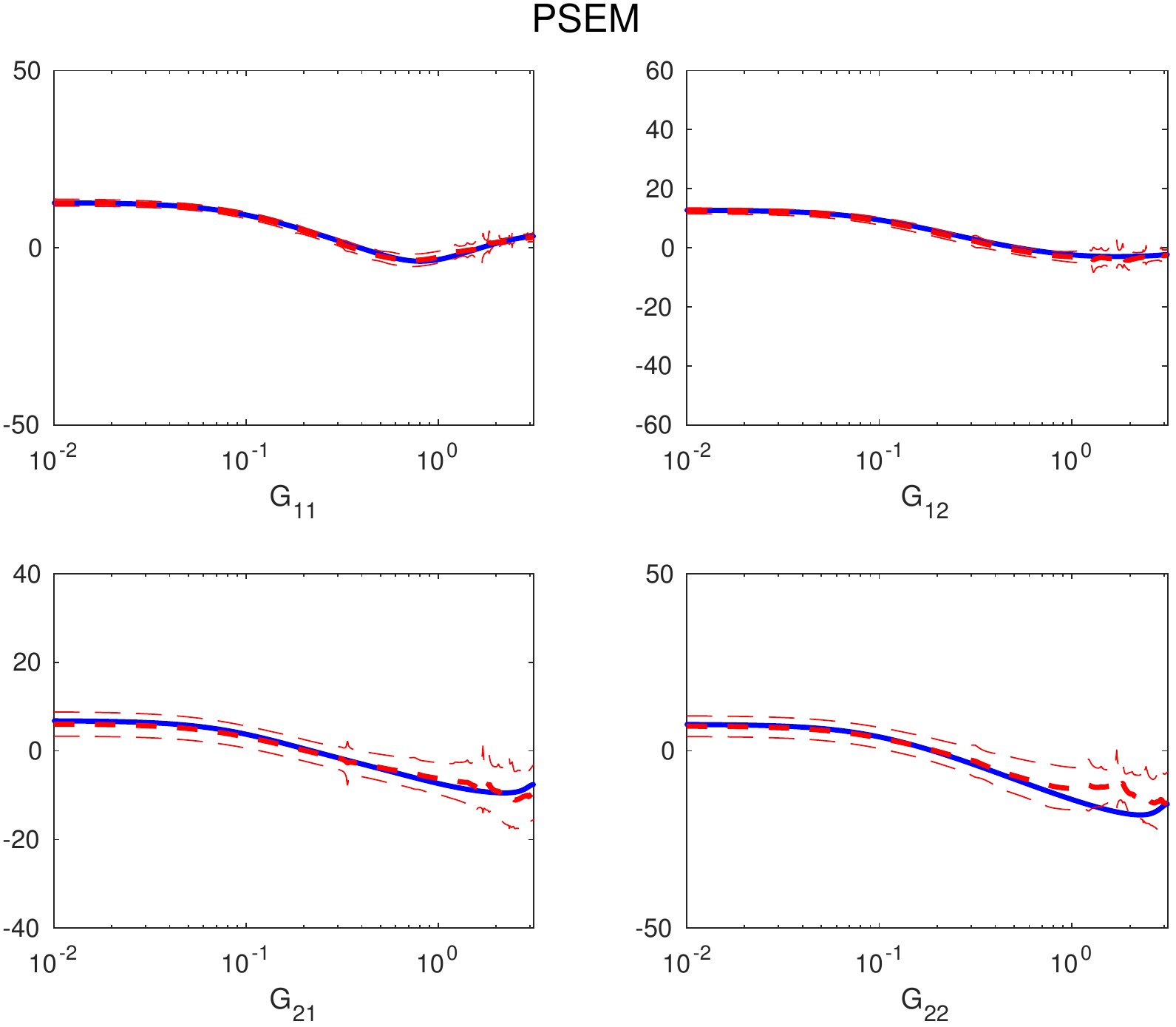}%
  \caption{Bode magnitude response using the PSEM method for
    the example studied in Section~\ref{sec:Ex5}.}
  \label{fig:bodemlem}
\end{center}
\end{figure}


\section{Conclusion and future work} \label{sec:conclusions}
In this paper we have presented a new quasi-Newton algorithm for
stochastic optimisation problems. The approach uses a tailored Gaussian process 
to model the unknown Hessian, which is then learned from gradient 
information as the algorithm progresses. To regulate the iterates, we
developed a stochastic line-search procedure that satisfies an Armijo condition (in
expectation) for early iterations, and converges to a deterministic
step-length schedule as the iterations grow. The former provides a mechanism
to handle poorly scaled problems while the Hessian is learned, and the
latter mimics conditions required for convergence in the stochastic
setting. The resulting combination is demonstrated on several
challenging nonlinear system identification problems, with promising
results. The method can be straightforwardly extended to handle the
situation where---possibly noisy---Hessian matrices are also observed.

The GP construction assumes that the gradient is corrupted by additive
Gaussian noise, and this is potentially not suitable for certain
problem classes, particularly where the noise distribution has heavy
tails. This allows for an extension to other processes such as the
student-t process \cite{ShahWG:2014}. A further extension is for the
case of large-scale problems, which frequently occur within the Machine
Learning field. Some initial work in this direction can be found in \cite{WillsS:2018}.

\section*{Acknowledgements}
This research is financially supported by the Swedish Research Council via the projects 
\emph{NewLEADS - New Directions in Learning Dynamical Systems} (contract number: 621-2016-06079) 
and
\emph{Learning flexible models for nonlinear dynamics} (contract number: 2017-03807), 
and 
the Swedish Foundation for Strategic Research (SSF) via the project \emph{ASSEMBLE} (contract number: RIT15-0012).


\appendix
\section{The stochastic Armijo condition}
\label{app:StochWolfeCond1}

%
%
All expectation below are conditioned on the variables $\{x_k, s_{\ell_{k-2}},
\widehat{y}_{\ell_{k-2}}\}$. This conditioning is dropped from the 
notation in order to improve readability. The required expectation is 
\begin{align}
  \label{eq:AppSWexpectation2}
  \Exp{\widehat{f}(x_k + \alpha_k {p}_k) - \widehat{f}(x_k) -
            c\alpha_k {g}_k^{\Transp} {p}_k} \leq 0,
\end{align}
and recall that 
\begin{subequations}
\label{eq:511}
\begin{align}
  \widehat{f}(x_k) &= f(x_k) + e_k,\label{eq:2}\\
  g_k &= \nabla f(x_k) + v_k,\label{eq:3}\\
  p_k &= -B_k g_k\\
                   &= -B_k (\nabla f(x_k) + v_k).\label{eq:4}
\end{align}
\end{subequations}
Assume that $x_k \in \mathcal{X}$ and $x_k+\alpha_k p_k \in
\mathcal{X}$. Since $f$ is assumed twice continuously differentiable
on $\mathcal{X}$, then employing Taylor's theorem results in
\begin{align}
  \notag
  f(x_k + \alpha_k p_k) &= f(x_k) + \alpha_k p_k^\Transp \nabla f(x_k)\\
  \label{eq:6} 
	&+ \frac{\alpha_k^2}{2} p_k^\Transp
                          \nabla^2 f(\bar{x}_k) p_k 
\end{align}
for some $\bar{x}_k \in \mathcal{X}$. Hence, using (\ref{eq:5}) and
(\ref{eq:6}) then (\ref{eq:AppSWexpectation2}) becomes
\begin{align*}
\textnormal{E} \Bigl [ &{f}(x_k) - f(x_k) + e_{k+1} - e_k\\ 
&- \alpha_k (\nabla f(x_k) + v_k)^\Transp B_k\nabla f(x_k)\\
&+c\alpha_k (\nabla f(x_k) + v_k)^{\Transp} B_k (\nabla f(x_k) + v_k) \\
&+ \frac{\alpha_k^2}{2} p_k^\Transp
                          \nabla^2 f(\bar{x}_k) p_k
\Bigr ] \leq 0
\end{align*}
which reduces to
\begin{align*}
  &\alpha_k \Bigl ( (1-c)\nabla f^\Transp (x_k) B_k\nabla f(x_k) - c \Tr \{B_kR\}
  \Bigr )\\
&\qquad\leq - \textnormal{E} \Bigl [ \frac{\alpha_k^2}{2} p_k^\Transp
                          \nabla^2 f(\bar{x}_k) p_k \Bigr ].
\end{align*}
Let $\gamma=\nabla f^\Transp (x_k) B_k\nabla f(x_k)$ and $\beta=\Tr
\{B_kR\}$, then division on both sides by $\alpha_k > 0$ results in
\begin{align}
  \label{eq:1}
  \gamma - c (\gamma + \beta) \leq -\alpha_k \kappa,
\end{align}
where
$\kappa = \textnormal{E} [ 1/2 p_k^\Transp \nabla^2 f(\bar{x}_k)
p_k]$. Since both $B_k \succ 0$ and $R \succ 0$ then $\gamma > 0$, and
$\beta > 0$. Therefore, for $\alpha_k>0$ small enough, then
\eqref{eq:1} is satisfied for
\begin{align*}
  0 < c < \bar{c} = \frac{\gamma}{\gamma + \beta}.
\end{align*}

\end{document}